\documentclass[lettersize,journal]{IEEEtran}
\IEEEoverridecommandlockouts

\usepackage{amsmath,amsfonts}
\usepackage{algorithm}
\usepackage{algpseudocode}
\usepackage{array}
\usepackage[caption=false,font=normalsize,labelfont=sf,textfont=sf]{subfig}
\usepackage{textcomp}
\usepackage{stfloats}
\usepackage{url}
\usepackage{verbatim}
\usepackage{graphicx}
\usepackage{cite}
\def\BibTeX{{\rm B\kern-.05em{\sc i\kern-.025em b}\kern-.08em
    T\kern-.1667em\lower.7ex\hbox{E}\kern-.125emX}}
\usepackage{balance}

\usepackage{makecell}
\usepackage{subfig}  
\usepackage{amssymb}
\usepackage{mathrsfs}
\usepackage{amsthm}
\RequirePackage{epsf}
\usepackage{epstopdf}
\usepackage[printonlyused]{acronym}

\acrodef{C-V2X}{cellular vehicle-to-everything}
\acrodef{V2I}{vehicle-to-infrastructure}
\acrodef{V2V}{vehicle-to-vehicle}
\acrodef{V2P}{vehicle-to-pedestrian}
\acrodef{RB}{resource block}
\acrodef{NR}{new radio}
\acrodef{QoS}{quality-of-service}
\acrodef{SINR}{signal-to-interference-plus-noise ratio}
\acrodef{CSI}{channel state information}
\acrodef{BS}{base station}
\acrodef{TxV}{transmitting vehicle}
\acrodef{RxV}{receiving vehicle}
\acrodef{RSSI}{received signal strength indicator}
\acrodef{RSS}{received signal strength}
\acrodef{PDF}{probability distribution function}
\acrodef{RV}{random variable}
\acrodef{MSE}{mean square error}
\acrodef{HR}{hazard rate}
\acrodef{ITS}{intelligent transportation systems}
\acrodef{PSFCH}{physical sidelink feedback channel}
\acrodef{PUCCH}{physical uplink control channel}
\acrodef{RV}{random variable}
\acrodef{i.i.d.}{independent, identically distributed}
\acrodef{GMM}{Gaussian mixture model}
\acrodef{TVaR}{tail value at risk}
\acrodef{CRF}{conditional relative frequency}
\acrodef{ORF}{overall relative frequency}

\acrodef{ISAC}{integrated sensing and communication}
\acrodef{IOC}{inverse optimal control}
\acrodef{mmWave}{millimeter wave}
\acrodef{DoF}{degrees of freedom}
\acrodef{VRU}{vulnerable road user}
\acrodef{AV}{autonomous vehicle}
\acrodef{HD}{high-definition}
\acrodef{NCV}{non-connected vehicle}
\acrodef{CV}{connected vehicle}
\acrodef{AoA}{angle-of-arrival}
\acrodef{AoD}{angle-of-departure}
\acrodef{ULA}{uniform linear array}
\acrodef{EKF}{extended Kalman filtering}
\acrodef{SNR}{signal-to-noise ratio}
\acrodef{VU}{vehicular user}
\acrodef{RIS}{reconfigurable intelligent surface}
\acrodef{AWGN}{additive white Gaussian noise}
\acrodef{MLE}{maximum likelihood estimator}
\acrodef{RSU}{roadside unit}
\acrodef{RCS}{radar cross section}
\acrodef{AD}{autonomous driving}
\acrodef{CAV}{connected autonomous vehicle}
\acrodef{ECM}{electronic countermeasure}
\acrodef{CP}{cooperative perception}
\acrodef{MDP}{Markov decision process}
\acrodef{RL}{reinforcement learning}
\acrodef{STL}{signal temporal logic}
\acrodef{NN}{neural network}
\acrodef{DTCR}{deep temporal clustering representation}
\acrodef{PPO}{proximal policy optimization}
\acrodef{CAS}{communication aided sensing}
\acrodef{SAC}{sensing aided communication}
\acrodef{ADPAR}{angular-domain peak-to-average ratio}
\acrodef{CRB}{Cramér-Rao bound}

\newtheorem{theorem}{Theorem}

\newtheorem{lemma}{Lemma}

\newtheorem{definition}{Definition}
\newtheorem{remark}{Remark}

\usepackage{siunitx}
\begin{document}

\title{Analysis and Detection of RIS-based Spoofing in Integrated Sensing and Communication (ISAC)\vspace{-0.3cm}}

\author{\normalsize Tingyu Shui, \IEEEmembership{Graduate Student Member, IEEE}, Po-Heng Chou, \IEEEmembership{Member, IEEE}, Walid Saad, \IEEEmembership{Fellow, IEEE}, and Mingzhe Chen, \IEEEmembership{Senior Member, IEEE}\vspace{-1.15cm}
\thanks{T. Shui and W. Saad are with Bradley Department of Electrical and Computer Engineering, Virginia Tech, Alexandria, VA, 22305, USA, Emails: tygrady@vt.edu, walids@vt.edu.}
\thanks{Po-Heng Chou is with the Research Center for Information Technology Innovation, Academia Sinica, Taipei, 11529, Taiwan, Email: d00942015@ntu.edu.tw.}
\thanks{M. Chen is with the Department of Electrical and Computer Engineering and Frost Institute for Data Science and Computing, University of Miami, Coral Gables, FL, 33146, USA, Email: mingzhe.chen@miami.edu.}
}
\maketitle
\begin{abstract}
\Ac{ISAC} is a key feature of next-generation 6G wireless systems, allowing them to achieve high data rates and sensing accuracy. While prior research has primarily focused on addressing communication safety in \ac{ISAC} systems, the equally critical issue of sensing safety remains largely under-explored. In this paper, the possibility of spoofing the sensing function of \ac{ISAC} in vehicle networks is examined, whereby a malicious \ac{RIS} is deployed to compromise the sensing functionality of a \ac{RSU}. For this scenario, the requirements on the malicious \ac{RIS}' phase shifts design and number of reflecting elements are analyzed. Under such spoofing, the practical estimation bias of the \ac{VU}'s Doppler shift and \ac{AoD} for an arbitrary time slot is analytically derived. Moreover, from the attacker's view, a \ac{MDP} is formulated to optimize the \ac{RIS}'s phase shifts design. The goal of this \ac{MDP} is to generate complete and plausible fake trajectories by incorporating the concept of spatial-temporal consistency. To defend against this sensing spoofing attack, a \ac{STL}-based neuro-symbolic attack detection framework is proposed and shown to learn interoperable formulas for identifying spoofed trajectories. Simulation results show that, with a beam misalignment of only $3.1^{\circ}$, a spoofed velocity can be estimated with a significant deviation from $-9.9$~m/s to $50$~m/s, when the ground truth is $10$~m/s. Moreover, an \ac{AoD} estimation bias of up to $11^{\circ}$ can be induced by the attacker. Such spoofed sensing outcomes cause a $28\%$ reduction in the achievable rate of the \ac{VU} in the beam tracking application. Finally, the proposed attack detection framework achieves an accuracy of $74.17\%$ in identifying the \ac{RIS} attack, which is $53.34\%$ higher than that of a deep clustering-based benchmark.

\end{abstract}
\acresetall
\vspace{-15pt}
\section{Introduction}
\vspace{-3pt}
To realize future \ac{ITS}, communications between transportation agents (vehicle and/or infrastructure) are fundamental for safety-critical applications such as cooperative perception and autonomous driving. Particularly, a vehicular network will need both sensing and communication services of high reliability and quality
. One promising approach is \ac{ISAC} \cite{10489861}, which equips vehicle and infrastructure with simultaneous communication and sensing capability. 

\Ac{ISAC} enables the integration of communication and sensing functions in \ac{ITS} through the co-design of the hardware architecture, transmitted waveforms, and signal processing algorithms on \ac{CAV} and \ac{RSU} \cite{10944644}. Such joint design benefits the vehicular network via two distinct aspects: \ac{CAS} and \ac{SAC} \cite{10845869}. First, for \ac{CAS}, the sensing accuracy and coverage can be enhanced by leveraging inter-agent transmission of environment perception information. One example is cooperative perception in which vehicles can extend their perception coverage by sharing the local sensing outcomes via wireless communication \cite{9606821}. Correspondingly, \ac{SAC} aims to exploit the sensing outcome to improve communication performance. Particularly, the sensing outcome pertaining to communication users, e.g., high-mobility \acp{VU}, can assist the transceiver to achieve accurate beam alignment and meet desired communication \ac{QoS} at high-frequency bands \cite{9171304}. However, such highly coupled design can be vulnerable to malicious attacks due to the broadcast nature of wireless transmission \cite{9755276}. Given the dire consequences of faulty sensing information in both \ac{CAS} and \ac{SAC} applications, it is essential to design \ac{ISAC} systems that can still operate effectively under adversarial conditions.

Aligned with the dual \ac{ISAC} functionalities, the security concerns in an \ac{ISAC} system should also be addressed from two distinct perspectives: \emph{communication safety} and \emph{sensing safety}. Generally, works on \emph{communication safety} \cite{9199556, 9737364,10781436}
reconsider the basic attacks in wireless communication, e.g., eavesdropping and jamming, in the context of \ac{ISAC}. Therein, the goal is to ensure confidential transmission while considering the tight coupling of sensing and communication functionalities. However, the \emph{sensing safety} issue in \ac{ISAC} systems remains largely under-explored. Some preliminary works are presented in \cite{10856886, 10443321, 10575930, 10634199}, in which the objective is, from a defender’s perspective, to prevent attackers from accessing the sensing information of authorized targets. Specifically, inspired by \ac{ECM} techniques, these works \cite{10856886, 10443321, 10575930, 10634199} leverage a \ac{RIS} to enhance the sensing echoes received at the authorized \ac{ISAC} station while suppressing those received at the unauthorized \ac{ISAC} station. Although the exploitation of \ac{RIS} improves the \emph{sensing safety} of \ac{ISAC} systems, this strategy also raises critical questions: \emph{Can an \ac{RIS} be adversarially deployed by an attacker to compromise the sensing safety of \ac{ISAC} systems instead? If so, how would this malicious \ac{RIS} be designed, and what impact would it have on the sensing outcomes? Can we design \ac{ISAC} systems resilient to such an attack on sensing functionality?}

\vspace{-10pt}
\subsection{Prior Works}
\vspace{-2pt}
The issue of sensing safety has recently attracted interest in a number of prior works \cite{10856886, 10443321, 10575930, 10634199, 10659120,10623793,article,chen2023metawave,11050921}. Specifically, in \cite{10856886, 10443321, 10575930} the authors used an  \ac{RIS}, as a legitimate component, to enhance the sensing safety of \ac{ISAC} systems by providing confidential sensing service. The phase shifts and amplitudes of the \ac{RIS}, mounted on the target, are deliberately designed for directional radar stealth. As a result, the targets can only be detected by the authorized \ac{ISAC} station, while remaining invisible to unauthorized \ac{ISAC} station. In \cite{10634199}, the \ac{RIS} was further exploited to simultaneously conceal the target and even spoof the unauthorized \ac{ISAC} station by redirecting the detection signal to clutter. As a result, a deceptive \ac{AoA} sensing outcome is generated at the unauthorized \ac{ISAC} station.

Meanwhile, the authors in \cite{10659120,10623793,article,chen2023metawave,11050921} examined the potential of sensing attacks, which aim to generate faulty sensing outcomes at the authorized \ac{ISAC} station by manipulating the sensing echoes. In \cite{10659120}, a sensing-resistant jammer was designed to disrupt the legitimate transmission while concealing itself from being detected. A novel metric named \ac{ADPAR} is proposed to quantify the detectability of the jammer. Then, the achievable rate of legitimate transmission is minimized under \ac{ADPAR} constraints. However, compared with active jamming, passive attacks pose greater challenges. In \cite{10623793}, the authors examined the impact of ``DISCO" \ac{RIS}, i.e., a malicious \ac{RIS} with random and time-varying reflection properties, on the \ac{ISAC} system. Moreover, the work in \cite{10623793} provided a theoretical analysis on the lower bound of the \ac{SINR} received by the communication users under the considered attack. In addition, the works in \cite{article} and \cite{chen2023metawave} futher experimentally implemented backscatters tags for sensing attacks. In \cite{11050921}, the authors extended the scenario to multiple jamming attackers and further proposed an attack detection method. In particular, in \cite{11050921}, a two-stage transmission protocol is designed, where beam scanning is used for attack detection in the first stage and a robust jamming beamforming is designed in the second stage to mitigate the impact of sensing attack. 

While some prior works \cite{10856886, 10443321, 10575930, 10634199, 10659120,10623793,article,chen2023metawave,11050921} have examined the sensing safety issue in \ac{ISAC} systems, a number of challenges remain unaddressed and require further investigation. Firstly, the prior art generally quantifies the impact of sensing attack by lower-bound metrics such as the \ac{CRB}, which reflect the theoretical limit of sensing accuracy. However, this overlooks the practical estimation bias introduced by spoofing attacks, which is critical for downstream tasks like beam tracking and trajectory planing. Secondly, none of the existing studies has examined spoofed sensing outcomes from a trajectory-level perspective. In particular, the existing literature only focuses on the impact of sensing spoof within a single time slot. Thus, when a full trajectory composed of multiple time slots is given, such a temporally isolated anomaly can be easily detected due to its inconsistency with adjacent time slots. In contrast, if a spoofing attack is capable of generating a complete plausible trajectory, it can greatly increase the severity of its impact and make detection significantly more challenging. Thirdly, the prior art, such as in \cite{10856886, 10443321, 10575930, 10634199, 10659120,10623793,article,chen2023metawave,11050921}, does not develop efficient detection mechanisms for sensing attacks on \ac{ISAC} systems. In our preliminary work \cite{shui2025sensing}, we studied the feasibility of the \ac{RIS} sensing attack for a vehicular network in which the goal of the \ac{RIS} is to manipulate the sensing echo received by the \ac{RSU}. The practical estimation bias of the \ac{VU}'s Doppler shift and \ac{AoD} are analytically derived. However, our work in \cite{shui2025sensing} only focuses on the spoofed sensing outcomes within a single time slot and does not examine the possibility of the malicious \ac{RIS} to generate plausible trajectories. Furthermore, the framework in \cite{shui2025sensing} does not provide any solution for attack detection in such spoofing use cases.

\vspace{-10pt}
\subsection{Contributions}
\vspace{-2pt}
The main contribution of this paper is, thus, the analysis of the impact of malicious \ac{RIS} on the sensing safety of \ac{ISAC} systems and the development of a corresponding attack detection framework. We particularly investigate the feasibility of \ac{RIS} spoofing on the \ac{RSU}'s estimation of a \ac{VU}'s Doppler shift and \ac{AoD}. The necessary conditions for the \ac{RIS} to conduct the spoofing are derived and the practical estimation bias of the sensing outcomes are analyzed, from both slot-level and trajectory-level perspectives. Finally, a \ac{STL}-based neuro-symbolic attack detection framework is proposed to identify spoofed trajectories. To our best knowledge, \textit{this is the first work that studies the feasibility of sensing attacks launched by malicious \ac{RIS}, analyzes its impact on sensing outcomes, and proposes a corresponding attack detection approach}. Hence, our key contributions include:
\begin{itemize}
    \item We consider a novel \ac{RIS}-aided spoofing scenario in which the sensing safety of an \ac{ISAC} vehicular network is compromised. We derive the necessary conditions on the \ac{RIS}, i.e. time-varying phase shifts and number of reflecting elements, of such a spoofing attack. We derive its impact on the \ac{RSU}'s estimation bias of the \ac{VU}'s Doppler shift and \ac{AoD}. Specifically, the feasible spoofing frequencies set with respect to the \ac{VU}'s Doppler shift for each single sensing slot is derived. Moreover, we analytically show that, the \ac{MLE} of the \ac{AoD} will also be jeopardized under the proposed \ac{RIS} spoofing. 
    \item From the attacker's perspective, we formulate a \ac{MDP} optimization problem that allows the attacker to design the \ac{RIS}' phase shifts in a way to generate complete and plausible fake trajectories that are difficult to distinguish from regular, real-world ones. Particularly, we define and incorporate the concept of spatial-temporal consistency into the \ac{MDP}, serving as a regularization term to ensure that the spoofed trajectory adheres to real-world physical constraints across consecutive time slots. This makes it challenging for the \ac{RSU} to identify the spoofed trajectories and detect the existence of a malicious \ac{RIS}. Due to the state-dependent action space of the proposed \ac{MDP}, an action-masking enhanced \ac{PPO} framework is deployed on the attacker.
    \item From a defender's view, we propose an \ac{STL}-based neuro-symbolic learning framework that detects potential sensing attacks by investigating the sensed \ac{VU}'s trajectory. In particular, \ac{STL} is a formal language used to specify the temporal properties of time signal sequence. By leveraging the TLINet approach of \cite{TLINet}, we explicit learn the formulas that a real-world trajectory should satisfy, which is then used to identify abnormal spoofed trajectory. Finally, an interpretable sensing attack detection mechanism can be realized according to the learned formulas.
    \item Simulation results demonstrate that the \ac{RIS} spoofing attack significantly compromises the sensing accuracy and the communication performance of the \ac{VU}. Specifically, with a beam misalignment of only $3.1^{\circ}$, the estimation error of the \ac{VU}'s velocity can range from $-9.9$~m/s to $50$~m/s, when the ground truth is $10$~m/s, and, an \ac{AoD} estimation bias of up to $11^{\circ}$ can be induced. Due to the spoofed sensing outcomes, the achievable rate of the \ac{VU} will be decreased by $28\%$ in the beam tracking application. Moreover, compared to the deep clustering-based detection benchmark, our proposed framework improves the detection accuracy by $53.34\%$, achieving a final accuracy of $74.17\%$ in identifying the \ac{RIS} attack.
\end{itemize}

The rest of the paper is organized as follows. The system model is presented in Section \ref{Section II}. The impact of the \ac{RIS} spoofing on the sensing outcome in a single time slot is analyzed in Section \ref{Section III}. Section \ref{Section IV} extends the analysis to trajectory level. Section \ref{Section V} introduces the \ac{STL}-based neuro-symbolic attack detection framework. Section \ref{Section VI} presents the simulation results and Section \ref{Section VII} concludes the paper.

\vspace{-9pt}
\section{System Model}
\vspace{-3pt}
\label{Section II}
Consider an \ac{ISAC} system supported by a full-duplex \ac{mmWave} \ac{RSU} \cite{9171304}. The \ac{RSU} is equipped with a \ac{ULA} of $N_t$ transmit antennas and $N_r$ receive antennas to provide sensing and downlink communications to a single-antenna \ac{VU}. By estimating the \ac{VU}'s kinetic states from echo signals, the \ac{RSU} can support critical applications such as beam tracking and safety alert. Following the standard convention in \cite{9171304}, we consider a scenario in which the \ac{VU} moves along a multi-lane straight road parallel to the \ac{RSU}'s \ac{ULA}.\footnote{The extension to non-parallel cases is straightforward and can be done by rotating the coordinate system.} A two-dimensional Cartesian coordinate system is used with the \ac{RSU} located at its origin, as shown in Fig. \ref{System_Model}. We focus on a time period discretized into $K$ short time slots, i.e., sensing time slots, and assume that each time slot corresponds to a short interval $T$. During time slot $k$, the goal of the \ac{RSU} is to estimate the \ac{VU}'s coordinates $(x_{k},y_{k})$ and velocity $v_k$,\footnote{The velocity $v_k$ captures only the x-directional component. The y-axis component can be either estimated by deploying another \ac{ULA} aligned with the y-axis or inferred from the uplink information transmitted by the \ac{VU} \cite{10561505}.} which are assumed to be constant in a given time slot \cite{9947033}. Thus, the \ac{RSU} can obtain the \ac{VU}'s trajectory $\boldsymbol{S} = [\boldsymbol{s}_1, \ldots, \boldsymbol{s}_K]$ with $\boldsymbol{s}_k = \left[ x_{k},y_{k}, v_k\right]^T$. In this considered system, a malicious \ac{RIS} is deployed to compromise the \ac{RSU}'s sensing functionality by manipulating the echo signals. Without loss of generality, we next introduce the system model with respect to time slot $k$.

\vspace{-8pt}
\subsection{Signal Model}
\vspace{-3pt}
Let $q_k(t) \in \mathbb{C}$ be the \ac{RSU}'s transmitted \ac{ISAC} symbol at time $t$ during slot $k$ with unit power, i.e., $q_k(t)q_k^{*}(t) = 1$. Given the precoding vector $\boldsymbol{w}_k \in \mathbb{C}^{N_t \times 1}$ at the \ac{RSU}, the echo signal reflected by the \ac{VU} will be given by:
\begin{equation}
\label{echo_V}
    \boldsymbol{y}_{\textrm{V},k}(t) = \sqrt{P} \gamma_{\textrm{B}} \beta_{\textrm{V},k} e^{j 2 \pi \mu_{k} t} \boldsymbol{b}_B(\theta_{k}) \boldsymbol{a}^{H}_B(\theta_{k}) \boldsymbol{w}_k q_k(t - \tau_{k}),
\end{equation}
where $P$ is the transmit power of the \ac{RSU}, $\gamma_{\textrm{B}} = \sqrt{N_t N_r}$ is the array gain factor, and $\beta_{\textrm{V},k} = \sqrt{\frac{\lambda^2 \kappa_{\textrm{V}}}{64 \pi^3 d_{k}^4}} e^{\frac{j 4 \pi d_{k}}{\lambda}}$ combines the complex path gain of the \ac{RSU}-\ac{VU}-\ac{RSU} link and the \ac{RCS} $\kappa_{\textrm{V}}$ of the \ac{VU} \cite{9724202}. The distance between the \ac{RSU} and the \ac{VU} is $d_{k}$ and the carrier wavelength is $\lambda$. Meanwhile, the Doppler shift $\mu_{k}$, \ac{AoD} $\theta_{k}$, and double-path echo delay $\tau_{k}$ of the \ac{VU} are included in \eqref{echo_V}, where $\boldsymbol{a}_B(\theta)$ and $\boldsymbol{b}_B(\theta)$ are the steering vectors of the transmitting and receiving antennas, respectively, at the \ac{RSU}. By assuming a half-wavelength antenna space, $\boldsymbol{a}_B(\theta)$ and $\boldsymbol{b}_B(\theta)$ will be given by \cite{9171304}:
\vspace{-6pt}
\begin{equation}
\label{steering_t} 
    \boldsymbol{a}_B(\theta)= \frac{1}{\sqrt{N_t}} \left[1, e^{-j \pi \cos \theta}, \ldots, e^{-j \pi\left(N_t-1\right) \cos \theta}\right]^T,
\end{equation}    
\vspace{-6pt}
\begin{equation}
\label{steering_r} 
    \boldsymbol{b}_B(\theta)= \frac{1}{\sqrt{N_r}} \left[1, e^{-j \pi \cos \theta}, \ldots, e^{-j \pi\left(N_r-1\right) \cos \theta}\right]^T.
\end{equation}   

\begin{figure}[t]
	\centering
	\includegraphics[scale=0.37]{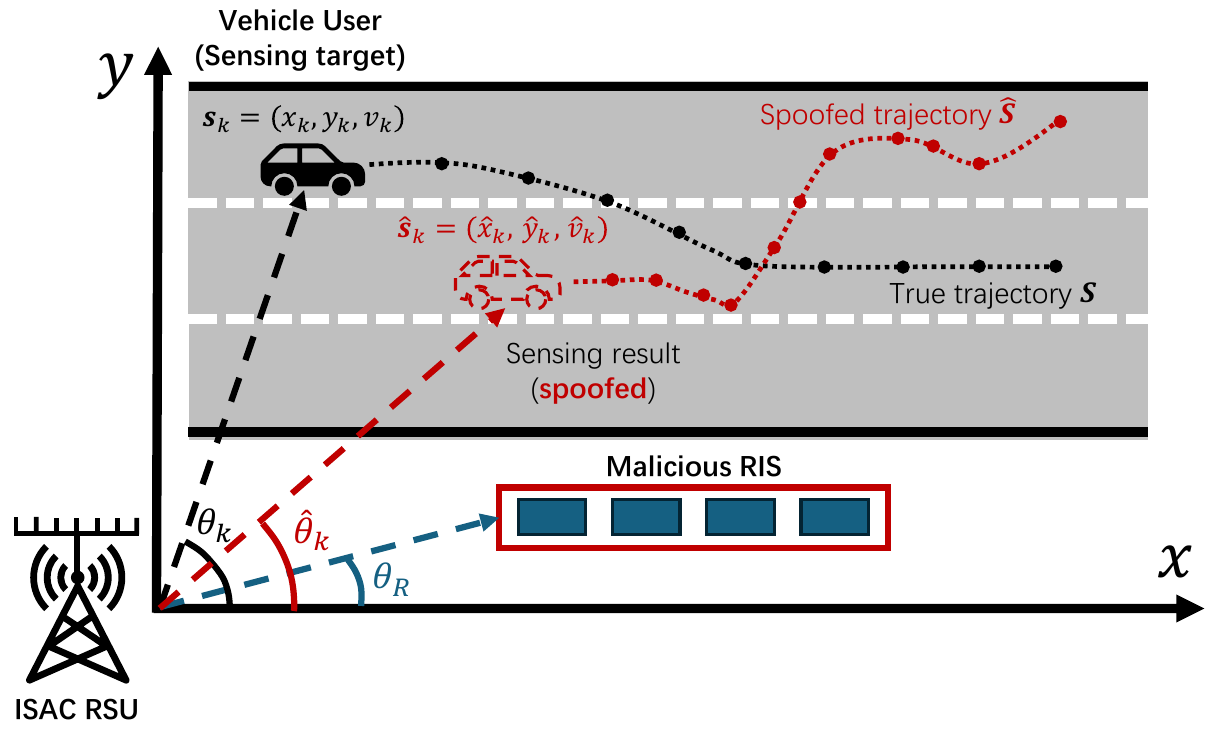}
      \vspace{-6pt}
	\caption{\small{Illustration of the considered \ac{ISAC} system model under the spoofing of a malicious \ac{RIS}.}}
 \label{System_Model}
   \vspace{-15pt}
\end{figure}
In presence of the malicious \ac{RIS}, the echo received by the \ac{RSU} is a composition of the legitimate echo reflected by the \ac{VU} and the spoofing echo adversarially reflected by the \ac{RIS}. We assume that the \ac{RIS} is equipped with a \ac{ULA} of $M$ elements. Thus, the spoofing echo from the \ac{RIS} will be:
\vspace{-4pt}
\begin{equation}
\label{echo_R}
\begin{aligned}
    \boldsymbol{y}_{\textrm{R},k}(t) 
    = & \sqrt{P} \gamma_{\textrm{B}} \beta_{\textrm{R}} \boldsymbol{b}_{\textrm{B}}(\theta_{\textrm{R}}) \boldsymbol{a}^{H}_{\textrm{R}}(\theta_{\textrm{B}}) \text{diag}[\boldsymbol{\phi}_k(t)] \times \\
    & \ \boldsymbol{b}_{\textrm{R}}(\theta_{\textrm{B}}) \boldsymbol{a}^{H}_{\textrm{B}}(\theta_{\textrm{R}}) \boldsymbol{w}_k q_k(t - \tau_{\textrm{R}}), \\
\end{aligned}
\end{equation}
where $\beta_{\textrm{R}} = \sqrt{\frac{\lambda^2 \kappa_{\textrm{R}}}{64 \pi^3 d_{\textrm{R}}^4}} e^{\frac{j 4 \pi d_{\textrm{R}}}{\lambda}}$ combines the complex path gain of the \ac{RSU}-\ac{RIS}-\ac{RSU} link and the \ac{RCS} $\kappa_{\textrm{R}}$ of the \ac{RIS}. According to \cite{10443321}, we have $\kappa_{\textrm{R}} = \frac{4 \pi \eta S^2}{\lambda^2}$ with the \ac{RIS}'s refection efficiency $\eta$, area $S$, and operating wavelength $\lambda$. Similarly, $\boldsymbol{b}_R(\cdot)$ and $\boldsymbol{a}_R(\cdot)$, $\theta_{\textrm{B}}$, and $\tau_{\textrm{R}}$ are the steering vectors, \ac{AoA} (and also \ac{AoD}), and echo delay of the \ac{RIS}. The phase shifts at the \ac{RIS} are given by $\boldsymbol{\phi}_k(t) = \left[ e^{j \phi_{k,1}(t)}, \ldots, e^{j \phi_{k,M}(t)}\right]$ and unit reflection amplitudes are assumed for simplicity. Note that $\boldsymbol{\phi}_k(t) $ must be time-varying to spoof the sensing process, as we will explain in Section \ref{Section III.A}. Thus, the composite echo received by the \ac{RSU} will be given by:
\vspace{-4pt}
\begin{equation}
\label{joint echo}
    \boldsymbol{y}_{k}(t) = \boldsymbol{y}_{\textrm{R},k}(t) + \boldsymbol{y}_{\textrm{V},k}(t) + \boldsymbol{z}_{\textrm{E}}(t),
\end{equation}
where $\boldsymbol{z}_{\textrm{E}}(t) \sim \mathcal{CN}(0, \sigma^2 \boldsymbol{I}_M)$ is the \ac{AWGN} at the \ac{RSU}'s receiving antennas.
\vspace{-8pt}
\subsection{Sensing Model}
\vspace{-2pt}
The goal of the sensing process on the \ac{RSU} during time slot $k$ is to estimate $\left[ \tau_{k}, \mu_{k}, \theta_{k} \right]$ for deriving $\boldsymbol{s}_k$. Particularly, the \ac{RSU} can determine $\mu_{k}$ and $\theta_{k}$ through a standard matched-filtering technique \cite{9171304}, after which the received signal in \eqref{joint echo} is compensated in both time and frequency domain for further estimating $\theta_{k}$. In particular, the matched-filtering output of $\boldsymbol{y}_{k}(t)$ is defined as follows:
\vspace{-6pt}
\begin{equation}
\vspace{-4pt}
\label{match filter}
C_k(\tau, \mu) \triangleq \left| \int_{0}^{T} \boldsymbol{y}_{k}(t) q^{*}(t - \tau) e^{-j 2 \pi \mu t} dt \right|^2.
\end{equation}
Thus, the estimated echo delay and Doppler shift can be given by $\left( \hat{\tau}_{k}, \ \hat{\mu}_{k} \right) = \arg \max_{\tau, \ \mu} C_k(\tau, \mu)$. Moreover, we consider the \ac{MLE} of $\theta_{k}$, defined as:
\vspace{-6pt}
\begin{equation}
\vspace{-4pt}
\label{MLE spoofed}
    \hat{\theta}_{k}=\arg \max _{\theta_{k}} p(\hat{\boldsymbol{y}}_{k} \mid \theta_{k}).
\end{equation}
In \eqref{MLE spoofed}, $p(\hat{\boldsymbol{y}}_{k} \mid \theta_{k})$ represents the likelihood function and $\hat{\boldsymbol{y}}_{k} = \int_{0}^{T} \boldsymbol{y}_{k}(t) q^{*}(t - \hat{\tau}_{k}) e^{-j 2 \pi \hat{\mu}_{k} t} dt$ is the echo compensated by $\hat{\tau}_{k}$ and $\hat{\mu}_{k}$. Given $\left[ \hat{\tau}_{k}, \hat{\mu}_{k},\hat{\theta}_{k}\right]$, we can estimate $\boldsymbol{s}_k$:
\vspace{-6pt}
\begin{equation}
\label{state}
    \hat{\boldsymbol{s}}_k = \left[ c \hat{\tau}_k \cos{\hat{\theta}_{k}}, c \hat{\tau}_k \sin{\hat{\theta}_{k}}, \frac{\hat{\mu}_{k} c}{f_c \cos{\hat{\theta}_{k}}}\right].
\end{equation}

The goal of the malicious \ac{RIS} is to hinder the \ac{RSU} from obtaining accurate $\left[ \hat{\tau}_{k}, \hat{\mu}_{k},\hat{\theta}_{k}\right]$, thereby degrading the accuracy of $\hat{\boldsymbol{s}}_k$. In other words, the spoofing echo $\boldsymbol{y}_{\textrm{R},k}(t)$ given in \eqref{echo_R} is expected to cause a deviation in the peak position of $C_k(\tau, \mu)$ and further distort the accuracy of the \ac{MLE} $\hat{\theta}_{k}$. In a typical \ac{ISAC} system with large bandwidth, the delay resolution of $C_k(\tau, \mu)$ is sufficiently high to resolve the two echoes given in \eqref{echo_V} and \eqref{echo_R}. Thus, the spoofed echo can be mitigated if the echo delay difference satisfies $|\tau_{k} - \tau_{\textrm{R}}| \geq \tau_0$, where $\tau_0$ is the effective main-lobe width of $C_k(\tau, \mu)$ in the time domain. To successfully conduct the echo spoofing, an adjustable-delay \ac{RIS} \cite{10423078} is assumed in our system, as shown in Fig. \ref{time-delay_RIS}. Specifically, a time-delay unit is cascaded before the \ac{RIS}'s phase shifter such that it is able to store and retrieve the impinging signals with a designed delay $\Delta t_{\textrm{R}}$. As a result, the echo delay in \eqref{echo_R} will be: $\tau_{\textrm{R}} = \Delta t_{\textrm{R}} + 2 \frac{d_{\textrm{R}}}{c}$. 
\vspace{-2pt}
\begin{remark}
    The reason for incorporating the adjustable-delay \ac{RIS} is to ensure that $|\tau_{k} - \tau_{\textrm{R}}| < \tau_0$. In other words, $\hat{\tau}_{k}$ is not spoofed. Hereinafter, we consider $|\tau_{k} - \tau_{\textrm{R}}| \approx 0$ since $\tau_0$, inversely proportional to the bandwidth \cite{yuan2024otfs}, is typically small in \ac{ISAC} systems (e.g., $2$~ns with bandwidth of $500$~MHz \cite{9947033}). Thus, a necessary condition for the \ac{RIS} echo spoofing is $ d_{\textrm{R}} \leq d_k \leq d_{\textrm{R}} + \frac{ \Delta_{\textrm{max}} c}{2} $ with $\Delta_{\textrm{max}}$ being the maximum adjustable-delay. For instance, given a typical value $\Delta_{\textrm{max}} = 0.32$~us \cite{10423078}, the spoofing range becomes $ d_{\textrm{R}} \leq d_k \leq d_{\textrm{R}} + 48~\textrm{m}$.
    \vspace{-2pt}
\end{remark}
Equipped with the time-delay unit, the malicious \ac{RIS} can manipulate the sensing echo received at the \ac{RSU}. Next, we derive the key conditions on the \ac{RIS}, i.e., its phase shifts design, for conducting the considered spoofing attack. Moreover, the spoofing attack's impact on the \ac{VU}'s sensing outcomes within time slot $k$ is analyzed.
\begin{figure}[t]
	\centering
	\includegraphics[scale=0.42]{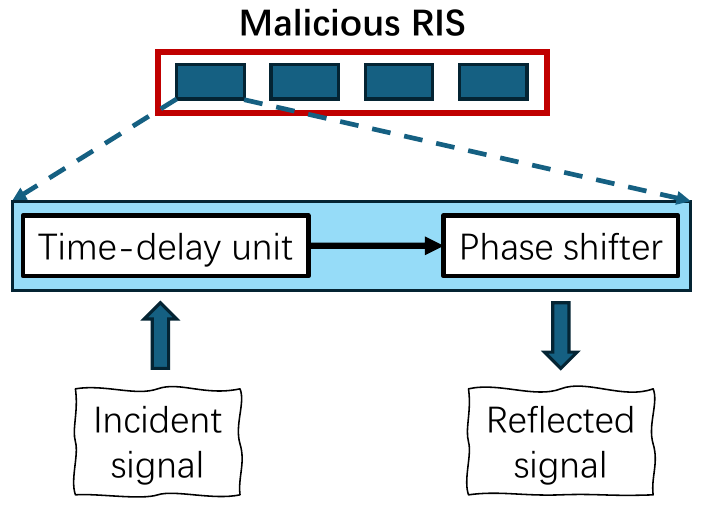}
     \vspace{-6pt}
	\caption{\small{The architecture of the adjustable-delay \ac{RIS}.}}
 \label{time-delay_RIS}
  \vspace{-15pt}
\end{figure}

\vspace{-10pt}
\section{Slot-Level Impact of Sensing Spoofing}
\vspace{-2pt}
\label{Section III}
In this section, we examine the conditions on the malicious \ac{RIS}'s phase shift design for a successful spoofing attack. From the slot-level perspective, we derive the feasible spoofing frequency set on the \ac{VU}'s Doppler shift estimation. Moreover, we further analyze the practical \ac{AoD} estimation bias.
\vspace{-8pt}
\subsection{Conditions on \ac{RIS}'s Phase shifts}
\vspace{-2pt}
\label{Section III.A}
First, we can rewrite the matched-filtering output in \eqref{match filter} as:
\vspace{-8pt}
\begin{equation}
\begin{aligned}
\label{match filtering c}
C_k(\tau, \mu) 
= & \Biggl| \sqrt{P} \gamma_{\textrm{B}}  \beta_{\textrm{V},k}  \boldsymbol{b}_B(\theta_{k}) \boldsymbol{a}^{H}_B(\theta_{k}) \boldsymbol{w}_k c_{\textrm{V},k}(\tau, \mu ) + \\
\vspace{-8pt}
& \ \sqrt{P} \gamma_{\textrm{B}}  \beta_{\textrm{R}} \boldsymbol{b}_{\textrm{B}}(\theta_{\textrm{R}}) \boldsymbol{a}^{H}_{\textrm{B}}(\theta_{\textrm{R}}) \boldsymbol{w}_k c_{\textrm{R},k}(\tau, \mu ) + \tilde{\boldsymbol{z}}_{\textrm{E}}(\tau, \mu)   \Biggl |^2, 
\end{aligned}
\end{equation}
where $c_{\textrm{V},k}(\tau, \mu ) \triangleq \int_{0}^{T} q(t - \tau_k)q^{*}(t - \tau) e^{ - j 2 \pi ( \mu -  \mu_{k}) t} dt $, $c_{\textrm{R},k}(\tau, \mu ) \triangleq \int_{0}^{T} q(t - \tau_k)q^{*}(t - \tau) \left[ \sum_{m=1}^{M} e^{ - j (2 \pi \mu t - \phi_{k,m}(t) )} \right] dt $, and $\tilde{\boldsymbol{z}}_{\textrm{E}}(\tau, \mu) \triangleq \int_{0}^{T} \boldsymbol{z}_{\textrm{E}}(t)q^{*}(t - \tau)e^{-j 2 \pi \mu t} dt $. To spoof the sensing outcome, we assume that the \ac{RIS} can derive $\Delta t_{\textrm{R}} = \tau_{\textrm{R}} -  2 \frac{ d_{\textrm{R}} }{c}$ by eavesdropping on the uplink communication of the \ac{VU}, which is a reasonable assumption \cite{9724202}. Thus, $C_k(\tau, \mu)$ will only exhibit one peak around $\tau_{k}$ in the time domain. Moreover, we assume a perfect echo delay estimation $\hat{\tau}_{k} = \tau_{k}$ due to the high delay resolution in large bandwidth \ac{mmWave} \ac{ISAC} systems. In other words, we only focus on the spoofed estimation of $\mu_k$ and $\theta_k$.

Given \eqref{match filtering c}, we can observe that the impact of \ac{RIS} spoofing stems from the term $\sqrt{P} \gamma_{\textrm{B}} \beta_{\textrm{R}} \boldsymbol{b}_{\textrm{B}}(\theta_{\textrm{R}}) \boldsymbol{a}^{H}_{\textrm{B}}(\theta_{\textrm{R}}) \boldsymbol{w}_k c_{\textrm{R},k}(\tau, \mu )$. Assuming $\hat{\tau}_{k} = \tau_{k}$, we have $c_{\textrm{V},k}(\hat{\tau}_{k}, \mu ) = \int_{0}^{T} e^{ - j 2 \pi ( \mu -  \mu_{k}) t} dt $ and $c_{\textrm{R},k}(\hat{\tau}_{k}, \mu ) = \sum_{m=1}^{M}  \int_{0}^{T} e^{ - j (2 \pi \mu t - \phi_{k,m}(t)) }dt $. We can find that the peak position of $c_{\textrm{V},k}(\hat{\tau}_{k}, \mu )$ in the frequency domain occurs at $\mu = \mu_{k}$. Thus, it is natural to set $\phi_{k,m}(t) = 2 \pi \tilde{\mu}_{k,m} t$ for artificially creating $M$ peak positions at frequencies $\tilde{\mu}_{k,1}, \ldots, \tilde{\mu}_{k,M}$. In practice, due to the hardware constraints of the \ac{RIS}, the phase shift $\phi_{k,m}(t)$ on the $m$-th element can be only given in a discrete form as:
\vspace{-6pt}
\begin{equation}
\label{phase shift}
 \phi_{k,m}(t) = (2 \pi \tilde{\mu}_{k,m} \lceil \frac{t}{\Delta T} \rceil \Delta T) \mod{2 \pi},
\end{equation}
where $\Delta T$ represents the shortest time interval over which $\phi_{k,m}(t)$ can vary. Meanwhile, the modulo operation is applied because the phase shifts of the \ac{RIS} are typically confined to the range $0$ to $2\pi$. Given the design in \eqref{phase shift}, we focus on one fundamental spoofing case in which the $M$ spoofing frequencies are set equal, i.e, $\tilde{\mu}_{k,1} = \ldots = \tilde{\mu}_{k,M} = \tilde{\mu}_k$ such that $\hat{\mu}_{k} \triangleq \arg \max_{\mu} C_k(\hat{\tau}_{k}, \mu)$ is spoofed as $\tilde{\mu}_k$. Note that, by setting all $M$ spoofing frequencies equal, the \ac{RIS} coordinates all its elements coherently to construct a peak of maximum magnitude at $C_k(\hat{\tau}_{k}, \tilde{\mu}_k)$, which constitutes a simple yet effective strategy to spoof the Doppler shift estimation. To this end, the \ac{RIS} should select $\tilde{\mu}_k$ such that the peak of $C_k(\hat{\tau}_{k}, \mu)$ occurs at $\tilde{\mu}_k$ rather than $\mu_{k}$, i.e., $C_k(\hat{\tau}_{k}, \tilde{\mu}_k) \geq  C_k(\hat{\tau}_{k}, \mu_{k})$. 

\vspace{-8pt}
\subsection{Impact on Sensing Outcomes}
\subsubsection{Doppler Shift Estimation}
We begin by deriving the feasible spoofing frequency set $\mathcal{A}_k \triangleq \{ \tilde{\mu}_k\mid C_k(\hat{\tau}_{k}, \tilde{\mu}_k) \geq  C_k(\hat{\tau}_{k}, \mu_{k}) \}$. Since the matched-filtering output of noise is negligible compared to the spoofing echo, we can ignore $\tilde{\boldsymbol{z}}_{\textrm{E}}(\hat{\tau}_{k}, \mu)$ and rewrite $C_k(\hat{\tau}_{k}, \mu)$ as:
\vspace{-6pt}
\begin{equation}
\label{match filter square}
\begin{aligned}
C_k(\hat{\tau}_{k}, \mu)
= & \frac{PT^2\gamma_{\textrm{B}}^2}{N_r} \sum_{n_r=1}^{N_r}  \Biggl|  \beta_{\textrm{V},k} g_{n_r}(\theta_{0,k}, \theta_{k}) e^{-j\pi(\mu-\mu_{k})T} \times \\
& \ \operatorname{sinc}(T (\mu - \mu_{k})) + \beta_{\textrm{R}} g_{n_r}(\theta_{0,k}, \theta_{\textrm{R}}) M e^{j \pi \tilde{\mu}_k\Delta T} \times \\
& \ e^{- j \pi (\mu - \tilde{\mu}_k)T} \operatorname{sinc}(\mu \Delta T) f(K,\pi \Delta T (\mu - \tilde{\mu}_k))\Biggr|^2,
\end{aligned}
\end{equation}
where we adopt a precoding vector $\boldsymbol{w}_k = \boldsymbol{a}_B(\theta_{0,k})$ steered towards $\theta_{0,k}$, which may be selected from a predefined codebook \cite{10740590}. Moreover, in \eqref{match filter square}, we define $f(K,x) = \frac{\sin{(Kx)}}{K \sin{x}}$, $\operatorname{sinc}(x) = \frac{\sin(\pi x)}{\pi x}$, and $g_{n_r}(\theta_1, \theta_2) = e^{-\frac{j \pi}{2} [ \cos{\theta_1(N_t -1)} - \cos{\theta_2} (N_t + 1 - 2n_r) ] } f(N_t, \frac{\pi}{2} (\cos{\theta_1} - \cos{\theta_{2}}) )$. To derive the expression of $\mathcal{A}_k$, we first present the following lemma pertaining to a set of infeasible spoofing frequencies.
\begin{lemma}
\label{lemma 1}
    If the number of \ac{RIS} reflecting elements satisfies $M \gg \sqrt{ \frac{\kappa_{\textrm{V}}}{4 \pi \eta}} \frac{\lambda}{S}  \left| \frac{f(N_t, \frac{\pi}{2} (\cos{\theta_{0,k}} - \cos{\theta_{k}}) )}{f(N_t, \frac{\pi}{2} (\cos{\theta_{0,k}} - \cos{\theta_{\textrm{R}}}) )} \right|$, an infeasible spoofing frequency set based on \eqref{match filter square} can be given by $ \mathcal{A}_{k,\varnothing} = \left\{ \mu \mid \mu = \mu_{k} + \frac{n}{\Delta T}, n \in \mathcal{Z} \right\}$, i.e., $\mathcal{A}_{k,\varnothing} \cap \mathcal{A}_k = \varnothing$.
\end{lemma}
\begin{proof}
See the proof in the conference version \cite{shui2025sensing}.
\end{proof}

Given Lemma 1, if we only consider $\tilde{\mu}_k\in \mathcal{A}_k$, i.e., $\tilde{\mu}_k\notin \mathcal{A}_{k,\varnothing}$, $C_k(\hat{\tau}_{\textrm{V}}, \mu)$ in \eqref{match filter square} can be approximated by:
\vspace{-6pt}
\begin{equation}
\label{match filter approx}
\begin{aligned}
C_k(\hat{\tau}_{k}, \mu) & \approx PT^2\gamma_{\textrm{B}}^2 \Biggl[ C_{\textrm{V},k} \operatorname{sinc}^2(T (\mu - \mu_{k})) + M^2 C_{\textrm{R}} \times \\
\vspace{-8pt}
& \quad  \operatorname{sinc}^2(\mu \Delta T)  f^2(K,\pi \Delta T (\mu - \tilde{\mu}_k)) \Biggr],
\end{aligned}
\vspace{-4pt}
\end{equation}
where the cross-product terms in the quadratic expansion are ignored since $\ f(K,\pi \Delta T (\mu - \tilde{\mu}_k))\operatorname{sinc}(T (\mu - \mu_{k})) \approx 0$, $\forall \mu$ when $\tilde{\mu}_k\notin \mathcal{A}_{k,\varnothing}$. In \eqref{match filter approx}, we further define $C_{\textrm{V},k} \triangleq \beta_{\textrm{V},k}^2 f^2(N_t, \frac{\pi}{2} (\cos{\theta_{0,k}} - \cos{\theta_{k}}) )$ and $C_\textrm{R} \triangleq \beta_{\textrm{R}}^2 f^2(N_t, \frac{\pi}{2} (\cos{\theta_{0,k}} - \cos{\theta_{\textrm{R}}}) )$. Another observation in \eqref{match filter approx} is that the impact of the spoofing frequency $\tilde{\mu}_k$ is periodic, as shown in the periodic spoofing term $M^2 C_\textrm{R} \operatorname{sinc}^2(\mu \Delta T)  f^2(K,\pi \Delta T (\mu - \tilde{\mu}_k))$. Specifically, the spoofing frequency $\tilde{\mu}_k$ will lead to multiple peaks at $\mu = \tilde{\mu}_k + \frac{n}{\Delta T}, n \in \mathcal{Z}$. Therefore, in order to derive $\mathcal{A}_k$, we only consider the highest peak defined in Lemma \ref{lemma 2}.
\begin{lemma}
\label{lemma 2}
    If the condition on $M$ in Lemma \ref{lemma 1} is satisfied, the highest peak of the spoofing term in \eqref{match filter approx} will be given by $ \Delta \tilde{\mu}_k\triangleq \tilde{\mu}_k\mod{ \frac{1}{\Delta T}}$.
\end{lemma}
\begin{proof}
See the proof in the conference version \cite{shui2025sensing}.
\end{proof}
\noindent
Given the result in Lemma \ref{lemma 2}, the possible spoofing frequency is restricted to the range $\left(0, \frac{1}{\Delta T}\right]$. In other words, an \ac{RIS} capable of changing its phase shift more frequently, i.e., operating with a smaller $\Delta T$, will have a broader range of spoofing frequencies $\Delta \tilde{\mu}$. Moreover, the feasible spoofing frequency set should be redefined as $\mathcal{A}_k \triangleq \left\{ \Delta \tilde{\mu}_k\mid C_k(\hat{\tau}_{k}, \Delta\tilde{\mu}_k) \geq  C_k(\hat{\tau}_{k}, \mu_{k})\right\}$, which is derived next.
\begin{theorem}
\label{theorem 1}
Under the condition on $M$ in Lemma \ref{lemma 1}, the feasible spoofing frequency set $\mathcal{A}_k$ is given in \eqref{feasible set}.
\end{theorem}
\begin{proof}
See the proof in the conference version \cite{shui2025sensing}.
\end{proof}

Theorem \ref{theorem 1} shows that, if the \ac{RIS} selects a spoofing frequency $\Delta \tilde{\mu}_k\in \mathcal{A}_k$, the estimated Doppler shift will be spoofed as $\hat{\mu}_{k} = \Delta \tilde{\mu}_k$. According to the derived result in \eqref{feasible set}, the range of feasible spoofing frequencies expands as more reflecting elements are deployed at the \ac{RIS}. In addition, the beam orientation $\theta_{0,k}$ also influences the set $\mathcal{A}_k$. Generally, when $\theta_{0,k}$ is closer to the \ac{VU}'s \ac{AoD} $\theta_{k}$, the spoofing set $\mathcal{A}_k$ becomes narrower, whereas a beam direction closer to the \ac{RIS}'s \ac{AoD} $\theta_{\textrm{R}}$ leads to a wider set. This phenomenon poses a critical challenge when \ac{ISAC} is leveraged for beam tracking. For instance, when a \ac{VU} moves through the coverage of the \ac{RSU}, there inevitably exists moments when the beam is steered closer to the \ac{RIS}. At such instances, the \ac{RIS} can select a spoofing frequency $\Delta \tilde{\mu}_k$ that results significant beam misalignment in the subsequent tracking step. More critically, this misalignment may accumulate over time, leading the \ac{RSU} to lose track of the \ac{VU} eventually and thereby jeopardizing the \ac{ISAC} sensing process.
\begin{figure*}[t]
\vspace{-15pt}
\begin{equation}
\vspace{-3pt}
    \label{feasible set}
    \mathcal{A}_k = \left\{ \Delta \tilde{\mu}_k \mid M^2 C_\textrm{R} \left[ \operatorname{sinc}^2(\Delta \tilde{\mu}_k\Delta T) - \operatorname{sinc}^2(\mu_{k} \Delta T)f^2(K,\pi \Delta T (\mu_{k} - \Delta \tilde{\mu}_k)) \right] - C_{\textrm{V},k} \left[ 1 - \operatorname{sinc}^2(T (\mu_{k} - \Delta \tilde{\mu}_k)) \right] \geq 0  \right\}.
\end{equation}
\vspace{-25pt}
\end{figure*}
\subsubsection{Impact on \ac{AoD} \ac{MLE}}
Given the estimated $\hat{\tau}_{k}$ and $\hat{\mu}_{k}$, the \ac{RSU} will continue to estimate $\theta_{k}$. First, the compensated and normalized echo is given by:
\vspace{-6pt}
\begin{equation}
\label{compensate signal}
\hat{\boldsymbol{y}}_{k} \triangleq \frac{1}{\sqrt{P} \gamma_{\textrm{B}}} \int_{0}^{T} \boldsymbol{y}_{k}(t) q^{*}(t - \hat{\tau}_k) e^{-j 2 \pi \Delta \tilde{\mu}_kt} dt.
\end{equation}
We next define a perfect \ac{MLE} $\tilde{\theta}_{k}$ obtained without spoofing and examine the deviation between the spoofed \ac{MLE} $\hat{\theta}_k$ and $\tilde{\theta}_{k}$. First, we rewrite \eqref{compensate signal} as:
\vspace{-6pt}
\begin{equation}
\begin{aligned}
\hat{\boldsymbol{y}}_{k}
= & \ T \beta_{\textrm{V},k}  \boldsymbol{b}_B(\theta_{k}) h(\theta_{k}, \theta_{0,k}) e^{-j\pi(\Delta \tilde{\mu}_k- \mu_{k})T} \times \\
& \quad\operatorname{sinc}(T (\Delta \tilde{\mu}_k- \mu_{k})) + M T \beta_{\textrm{R}} \boldsymbol{b}_{\textrm{B}}(\theta_{\textrm{R}}) \times \\
& \quad h(\theta_{\textrm{R}}, \theta_{0,k}) e^{j \pi \Delta \tilde{\mu}_k\Delta T} \operatorname{sinc}(\Delta \tilde{\mu}_k\Delta T) + \hat{\boldsymbol{z}}_{\textrm{E}},
\end{aligned}
\end{equation}
where $\hat{\boldsymbol{z}}_{\textrm{E}} = \frac{\tilde{\boldsymbol{z}}_{\textrm{E}}(\hat{\tau}, \Delta \tilde{\mu}_k) }{\sqrt{P}\gamma_{\textrm{B}}} \sim \mathcal{C N}\left(\boldsymbol{0}_{N_r}, \frac{\sigma^2T}{P \gamma_{\textrm{B}}^2} \boldsymbol{I}_{N_r}\right)$ and $h(\theta_1, \theta_2) = e^{-\frac{j \pi(N_t -1)}{2} [ \cos{\theta_1} - \cos{\theta_2} ] } f(N_t, \frac{\pi}{2} (\cos{\theta_1} - \cos{\theta_{2}}) )$. Similar to \eqref{MLE spoofed}, we define the perfect \ac{MLE} as follows:
\vspace{-6pt}
\begin{equation}
\label{MLE perfect}
    \tilde{\theta}_{k}=\arg \max _{\theta_{k}} p( \tilde{\boldsymbol{y}}_k \mid \theta_{k}),
\end{equation}
where $\tilde{\boldsymbol{y}}_k = T \beta_{\textrm{V},k} \boldsymbol{b}_B(\theta_{k}) h(\theta_{k}, \theta_{0,k}) + \hat{\boldsymbol{z}}_{\textrm{E}}$. Here, $\tilde{\boldsymbol{y}}_k$ represents the normalized received signal compensated by perfect estimation $\hat{\tau}_{k} = \tau_{k}$ and $\hat{\mu}_{k} = \mu_{k}$, without the \ac{RIS} spoof. Thus, $\tilde{\theta}_k$ in \eqref{MLE perfect} is the best \ac{MLE} on $\theta_{k}$ we can obtain. 

Next, we derive the deviation between the perfect \ac{MLE} $\tilde{\theta}_{k}$ and the spoofed \ac{MLE} $\hat{\theta}_{k}$ in the following theorem.
\begin{theorem}
\label{theorem 2}
Assume the perfect delay estimation $\hat{\tau}_k = \tau_{k}$ and spoofed Doppler shift estimation $\hat{\mu}_{k} = \Delta \tilde{\mu}$ are obtained under \ac{RIS} spoofing. The spoofed \ac{MLE} $\hat{\theta}_{k}$ can be given by:
\vspace{-6pt}
\begin{equation}
\label{MLE spoofed analytical in theorem}
\begin{aligned}
\hat{\theta}_{k} = & \arg \min _{\theta_{k}} \Biggl [  \left\| \tilde{\boldsymbol{y}}_k - T \beta_{\textrm{V},k} \boldsymbol{b}_B(\theta_{k}) h(\theta_{k}, \theta_{0,k})\right\|^2 +  \\
\vspace{-4pt}
& \ 2 T \beta_{\textrm{V},k} \sum_{n=1}^{N_r} \Re{\left\{ \Delta y_{k,n} g_{n}^{*}(\theta_{k}, \theta_{0,k}) \right\}} \Biggr ],
\end{aligned}
\end{equation}
where $\Delta \boldsymbol{y}_k \triangleq \hat{\boldsymbol{y}}_{k} - \tilde{\boldsymbol{y}}_k = \left[ \Delta  y_{k,1}, \ldots,  \Delta  y_{k,N_r} \right]^{T}$. Moreover, the perfect \ac{MLE} $\hat{\theta}_{k}$ defined in \eqref{MLE perfect} can be similarly given by: 
\vspace{-6pt}
\begin{equation}
\label{MLE perfect analytical}
\begin{aligned}
\tilde{\theta}_{k} = \arg \min_{\theta_{k}} \left\|\tilde{\boldsymbol{y}}_k - T\beta_{\textrm{V},k} \boldsymbol{b}_B(\theta_{k}) h(\theta_{k}, \theta_{0,k})\right\|^2.
\end{aligned}
\end{equation}
\end{theorem}
\begin{proof}
See the proof in the conference version \cite{shui2025sensing}.
\end{proof}
\noindent
From Theorem \ref{theorem 2}, we can observe that the spoofed \ac{MLE} deviates from the perfect \ac{MLE} because of the term $2 T \beta_{\textrm{V},k} \sum_{n=1}^{N_r} \Re{\left\{ \Delta y_{k,n} g_{n}^{*}(\theta_{k}, \theta_{0,k}) \right\}}$, which is introduced by the malicious \ac{RIS}. Although it is challenging to derive a closed-form expression for $\left| \tilde{\theta}_{k} - \hat{\theta}_{k}\right|$, we can still find that a spoofed \ac{MLE} $ \hat{\theta}_{k} \neq \tilde{\theta}_{k}$ will be obtained if the location of the minimum of $\left\| \hat{\boldsymbol{y}} - T \beta_{\textrm{V},k} \boldsymbol{b}_B(\theta_{k}) h(\theta_{k}, \theta_{0,k})\right\|^2$ is altered due to the addition of the term $2 T \beta_{\textrm{V},k} \sum_{n=1}^{N_r} \Re{\left\{ \Delta y_{k,n} g_{n}^{*}(\theta_{k}, \theta_{0,k}) \right\}}$. To further illustrate the impact of \ac{RIS} spoofing on the \ac{AoD} estimation, we derive $\Delta \boldsymbol{y}_k$ as follows:
\vspace{-4pt}
\begin{equation}
\label{Delta y}
\begin{aligned}
\Delta \boldsymbol{y}_{k}
= & \ T \beta_{\textrm{V},k}  \boldsymbol{b}_B(\theta_{k}) h(\theta_{k}, \theta_{0,k}) \times \\ 
& \quad \left[ e^{-j\pi(\Delta \tilde{\mu}_k- \mu_{k})T} \operatorname{sinc}(T (\Delta \tilde{\mu}_k- \mu_{k})) - 1\right] \\
& + M T \beta_{\textrm{R}} \boldsymbol{b}_{\textrm{B}}(\theta_{\textrm{R}}) h(\theta_{\textrm{R}}, \theta_{0,k}) e^{j \pi \Delta \tilde{\mu}_k\Delta T} \operatorname{sinc}(\Delta \tilde{\mu}_k\Delta T).
\end{aligned}
\end{equation}
From \eqref{Delta y}, we can observe that the \ac{RIS}'s impact on the \ac{AoD} estimation stems from the second term. Particularly, given the same spoofing frequency $\Delta \tilde{\mu}$, an \ac{RIS} with a larger number of reflecting elements $M$, higher path gain $\beta_{\textrm{R}}$, and an \ac{AoD} $\theta_{\textrm{R}}$ closer to the \ac{RSU}'s beam direction $\theta_{0,k}$ tends to induce more significant \ac{AoD} estimation bias. 

Theorems \ref{theorem 1} and \ref{theorem 2} explicitly demonstrate estimation bias of $\hat{\theta}_{k}$ and $\hat{\mu}_k$ for an arbitrary time slot $k$. However, when a full trajectory composed of multiple time slots is considered, a single spoofed sensing outcome can be easily detected due to its inconsistency with adjacent time slots. Therefore, from the attacker's view, we next examine the feasibility of such \ac{RIS} spoofing from a trajectory-level perspective, with the goal of generating complete and plausible fake trajectories.

\vspace{-10pt}
\section{Trajectory-level Impact of Sensing Spoofing}
\label{Section IV}
In this section, we further investigate the phase shift design of the \ac{RIS} for spoofing a complete trajectory at the \ac{RSU}. Specifically, from the attacker’s perspective, the objective is to generate a plausible trajectory that satisfies real-world spatial-temporal consistency, thereby increasing the difficulty of attack detection. To this end, we formulate an \ac{MDP} for optimizing the spoofing frequency sequence on the \ac{RIS}, and solve it using a \ac{PPO} framework enhanced with action mask.

\vspace{-8pt}
\subsection{Problem Formulation}
The aim of the malicious \ac{RIS} is to select the spoofing frequency to induce spoofed trajectory $\hat{\boldsymbol{S}} = [\hat{\boldsymbol{s}}_1, \ldots, \hat{\boldsymbol{s}}_K]$ at the \ac{RSU}, with $\hat{\boldsymbol{s}}_k = \left[ \hat{x}_k, \hat{y}_k, \hat{v}_k \right]$. In practice, the regular trajectory of a \ac{VU} should follow certain spatial and temporal consistency. Thus, to generate plausible sensing outcomes that cannot be easily identified, the spoofing frequency sequence $ \Delta \boldsymbol{\tilde{\mu}} = \left[ \Delta\tilde{\mu}_1, \ldots, \Delta\tilde{\mu}_K \right]$ should be designed carefully.

We begin by defining the spatial and temporal consistency of a trajectory. Formally, a trajectory $\boldsymbol{S} = [\boldsymbol{s}_1, \ldots, \boldsymbol{s}_K]$ with $\boldsymbol{s}_k \in \mathbb{R}^m$ follows spatial and temporal consistency when a vector-valued constraint $\boldsymbol{c}(\boldsymbol{s}_k, \boldsymbol{s}_{k+1}) : \mathbb{R}^m \times \mathbb{R}^m \to \mathbb{R}^n$ satisfies $\boldsymbol{c}(\boldsymbol{s}_k, \boldsymbol{s}_{k+1}) \succeq \mathbf{0}$ for all $k \in \left\{ 1, \ldots, K-1 \right\}$.
 Here, the expression of constraint $\boldsymbol{c}(\cdot, \cdot)$ and dimension $n$ are both related to the physical law that a real-world trajectory should follow. For instance, a \ac{VU} is subject to physical constraints on acceleration. Thus, the constraint $\boldsymbol{c}(\boldsymbol{s}_{k}, \boldsymbol{s}_{k+1}) \succeq \mathbf{0}$ may encode condition $v_{k+1} - v_k \leq a_{\textrm{max}} T$, where $a_{\textrm{max}}$ is the maximum acceleration. To formulate a comprehensive $\boldsymbol{c}(\boldsymbol{s}_k, \boldsymbol{s}_{k+1})$ is challenging since the physical constraints in real world can be complicated. For tractability, we consider a simplified but practical consistency:
\vspace{-6pt}
\begin{equation}
\label{consistency}
\boldsymbol{c}(\boldsymbol{s}_k, \boldsymbol{s}_{k+1})  =  
\left[
\begin{array}{c}
a_{\textrm{max}} T - (v_{k+1} - v_k) \\
(v_{k+1} - v_k) - a_{\textrm{min}} T \\
\delta_x - \left\| x_{k+1} - x_{k} -v_k T \right\|\\
\delta_y - \left\| y_{k+1} - y_{k}\right\|
\end{array}
\right],
\end{equation}
where $a_{\textrm{min}}$ is the minimum acceleration (maximum deceleration), $\delta_x$ is the estimation error margin along x-axis, and $\delta_y$ is the estimation error margin along y-axis. The first two terms in \eqref{consistency} capture the change in the \ac{VU}'s velocity between adjacent steps, while the last two terms represent the change in the \ac{VU}'s location. Thus, we can formulate the spoofing frequency design as a feasibility problem given by:
\vspace{-6pt}
\begin{subequations}
\label{opt1}
\begin{IEEEeqnarray}{s,rCl'rCl'rCl}
& \text{find} &\quad& \Delta \boldsymbol{\tilde{\mu}} \in \mathbb{R}^T \label{obj1}\\
&\text{s.t.} && \Delta\tilde{\mu}_k \in \mathcal{A}_{k}, \forall k = 1, \ldots, K,  \label{c1-1} \\
&&& \boldsymbol{c}(\hat{\boldsymbol{s}}_k, \hat{\boldsymbol{s}}_{k+1}) \succeq \mathbf{0}, \forall k = 0, \ldots, K-1, \label{c1-2}
\end{IEEEeqnarray}
\end{subequations}
where \eqref{c1-1} indicates that the spoofing frequency at each time slot should be selected from the corresponding feasible spoofing frequency set and \eqref{c1-2} ensures that the spatial temporal consistency is satisfied for the whole trajectory. Moreover, we assume an accurate initial sensing, i.e., $\hat{\boldsymbol{s}}_0 = \boldsymbol{s}_0$.

\vspace{-8pt}
\subsection{\Ac{MDP} with State-dependent Action Mask}
The feasibility problem \eqref{opt1} can be interpreted as a sequential decision-making process. Specifically, the \ac{RIS} needs to determine its spoofing frequency $\Delta \hat{\mu}_k$ at time slot $k$ to generate spoofed sensing outcome $\hat{\boldsymbol{s}}_k$. Moreover, the consistency constraint given in \eqref{c1-2} should be satisfied to ensure a plausible trajectory at the \ac{RSU}. Note that the \ac{RIS} only needs to know the true state $\boldsymbol{s}_k$ of the \ac{VU} to perform such spoofing (as detailed later in this section), and this state can be easily obtained through the uplink transmission of the \ac{VU} \cite{9724202}. To solve \eqref{opt1}, we next map it into a an \ac{MDP} and propose a \ac{RL}-based approach.

First, we establish the first-order Markov property of the sequential decision-making process by defining its action $\Delta \tilde{\mu}$ and state $\boldsymbol{\xi}_{k} = [\hat{\boldsymbol{s}}_{k-1}, \boldsymbol{s}_k]$, where $\hat{\boldsymbol{s}}_{k-1}$ is the spoofed sensing outcome at time slot $k-1$ and $\boldsymbol{s}_{k}$ is the true state of \ac{VU} at time slot $k$. Our goal is to show that the future state $\boldsymbol{\xi}_{k+1}$ depends only on the current state $\boldsymbol{\xi}_{k}$ and action $\Delta\tilde{\mu}_k$, and is conditionally independent of past states and actions. Recall that the spoofed sensing outcome under $\Delta\tilde{\mu}_k \in \mathcal{A}_k$ in \eqref{state} is given as $\hat{\boldsymbol{s}}_k = \left[ c \tau_k \cos{\hat{\theta}_{k}}, c \tau_k \sin{\hat{\theta}_{k}}, \frac{\Delta\tilde{\mu}_k c}{f_c \cos{\hat{\theta}_{k}}}\right]$, where $\hat{\theta}_{k}$ is determined by the sensing outcome $\hat{\boldsymbol{s}}_{k-1}$, the true state $\boldsymbol{s}_{k}$ of \ac{VU}, and the beam direction $\theta_{0,k}$ according to Theorem \ref{theorem 2}. Thus, we can view $\hat{\boldsymbol{s}}_k$ as a function $\hat{\boldsymbol{s}}_k(\hat{\boldsymbol{s}}_{k-1}, \boldsymbol{s}_{k}, \Delta\tilde{\mu}_k, \theta_{0,k})$. Based on \cite{9171304}, we can further derive that 
\vspace{-5pt}
\begin{equation}
\label{beam}
    \theta_{0,k} \overset{(a)}{=} \theta_{k|k-1} \overset{(b)}{\approx} \hat{\theta}_{k-1} - \frac{\hat{v}_{k-1} \sin{\hat{\theta}_{k-1}}T }{\sqrt{x_{k-1}^2 + y_{k-1}^2}},
    \vspace{-1pt}
\end{equation}
where $\theta_{k|k-1}$ is the predicted \ac{AoD} at time slot $k$, given the estimated \ac{AoD} $\hat{\theta}_{k-1}$ at time slot $k-1$. Specifically, (a) in \eqref{beam} captures the principle of beam tracking, whereby the beam direction $\theta_{0,k}$ at time slot $k$ is determined based on the previously estimated \ac{AoD} of the \ac{VU} at time slot $k-1$. Moreover, term (b) in \eqref{beam} is derived based on geometric relationships, under the assumption that the duration of a single time slot is sufficiently short \cite{9171304,10561505,9947033}. As a result, we can view $\theta_{0,k}$ as a function $\theta_{0,k}(\hat{\boldsymbol{s}}_{k-1})$. Now we can rewrite $\hat{\boldsymbol{s}}_k(\hat{\boldsymbol{s}}_{k-1}, \boldsymbol{s}_{k}, \Delta\tilde{\mu}_k, \theta_{0,k})$ as $\hat{\boldsymbol{s}}_k(\boldsymbol{\xi}_k, \Delta\tilde{\mu}_k)$. Finally, we can conclude that the process satisfies the first-order Markov property since $\boldsymbol{\xi}_{k+1} = [\hat{\boldsymbol{s}}_{k}, \boldsymbol{s}_{k+1}] = [\hat{\boldsymbol{s}}_{k}(\boldsymbol{\xi}_k, \Delta\tilde{\mu}_k), \boldsymbol{s}_{k+1}]$ and $\boldsymbol{s}_k$, determined only by the \ac{VU}, is independent of all past states and actions in the process.

Next, we map problem \eqref{opt1} into an \ac{MDP} $\mathcal{M} = \left\{ \Xi, \mathcal{A}, \mathcal{P}, R, \gamma\right\}$, where $\Xi$ is the state space, $\mathcal{A}$ is the action space, $\mathcal{P}: \Xi \times \Xi \times \mathcal{A} \to \mathbb{R}$ is the stochastic state transition function,
$R: \Xi \times \mathcal{A} \to \mathbb{R}$ is the reward function, and $\gamma$ is the discount factor for future reward. The components of $\mathcal{M}$ is detailed as the following:
\begin{itemize}
    \item \textit{Agent}: The agent is the attacker that controls the malicious \ac{RIS}.
    \item \textit{Action}: According to Theorem \eqref{theorem 1} and analysis in Section \ref{Section III}, the action at time slot $k$ is defined as the spoofing frequency, i.e., $\Delta\tilde{\mu}_k \in \mathcal{A}_k \subset \mathcal{A}$. In other words, the actual action space $\mathcal{A}_k$ at time slot $k$ is a subspace of $\mathcal{A}$, which varies across different time slots.
    \item \textit{States}: The state $\boldsymbol{\xi}_k$ at time slot $k$ is the spoofed sensing outcome $\hat{\boldsymbol{s}}_{k-1}$ at last time slot $k-1$ and the true state $\boldsymbol{s}_{k}$ of \ac{VU} at time slot $k$, i.e., $\boldsymbol{\xi}_{k} = [\hat{\boldsymbol{s}}_{k-1}, \boldsymbol{s}_k] \in \Xi$.
    \item \textit{Reward}: The reward assigned to the agent at time slot $k$ must reflect how well the spoofed sensing outcome $\hat{\boldsymbol{s}}_k$ satisfies the spatial and temporal consistency with respect to $\hat{\boldsymbol{s}}_{k-1}$. Thus, we define $R(\boldsymbol{\xi}_k, \Delta\tilde{\mu}_k) = \sum_{i} \min(0, c_i(\hat{\boldsymbol{s}}_{k-1}, \hat{\boldsymbol{s}}_{k})) $, where $c_i(\hat{\boldsymbol{s}}_{k-1}, \hat{\boldsymbol{s}}_{k})$ is the $i$-th element in $\boldsymbol{c}(\hat{\boldsymbol{s}}_{k-1}, \hat{\boldsymbol{s}}_{k})$ \footnote{The reward $R(\boldsymbol{\xi}_k, \Delta\tilde{\mu}_k)$ is accessible to the \ac{RIS} as both $\hat{\boldsymbol{s}}_{k-1}$ and $\hat{\boldsymbol{s}}_{k}$ can be derived by the \ac{RIS} through Theorems \ref{theorem 1} and \ref{theorem 2}.}. Specifically, when $c_i(\hat{\boldsymbol{s}}_{k-1}, \hat{\boldsymbol{s}}_{k}) \geq 0$ for $\forall i$ is satisfied, the reward is maximized as $R(\boldsymbol{\xi}_k, \Delta\tilde{\mu}_k) = 0$, which captures the condition that the spoofed sensing outcome satisfies the spatial and temporal consistency and is difficult to distinguish from regular and real-world trajectories. 
\end{itemize}

To solve the sequential decision-making problem in $\mathcal{M}$, we can use any off-the-shelf \ac{RL} framework to find the policy $\pi$ that maximizes the expected discounted reward $J = \mathbb{E}_{\pi} \left[ \sum_{k=0}^{\infty} \gamma^k R(\boldsymbol{\xi}_k, \Delta\tilde{\mu}_k) \right]$. However, one challenge here is that the feasible action of $\mathcal{M}$ depends on the specific state. Particularly, Theorem \ref{theorem 1} shows that $\mathcal{A}_k$ can be viewed as a function $\mathcal{A}_k\left(\theta_{0,k}, \boldsymbol{s}_k \right)$, while other parameters are time-invariant. Since $\theta_{0,k}$ can be viewed as $\theta_{0,k}(\hat{\boldsymbol{s}}_{k-1})$, we can rewrite $\mathcal{A}_k\left(\theta_{0,k}, \boldsymbol{s}_k \right)$ as $\mathcal{A}_k\left(\hat{\boldsymbol{s}}_{k-1}, \boldsymbol{s}_k\right) = \mathcal{A}_k\left(\boldsymbol{\xi}_k\right)$, i.e., the feasible action set of $\mathcal{M}$ varies with its current state. To discard the infeasible actions, we next incorporate a popular \ac{RL} technique named action mask \cite{huang2020closer}.

To use action mask, we first discretize $\mathcal{A} = [0, \frac{1}{\Delta T}]$ into $L$ discrete actions as $\hat{\mathcal{A}} = \left\{ \Delta \mu, \ldots, L\Delta \mu  \right\}$ with the discretization step $\Delta \mu  = \frac{1}{\Delta T L}$. As a result, we obtain a \ac{MDP} of discrete action space and continuous state space. The standard \ac{PPO} framework is selected for solving the \ac{MDP}, as it supports discrete action spaces and ensures stable training through clipped policy updates. Moreover, as a policy-based \ac{RL} framework, \ac{PPO} directly outputs the probabilities of selecting certain actions, which can be seamlessly integrated with the action masking mechanism to satisfy the state-dependent action constraints \eqref{c1-1}. Specifically, we seek to train a policy network of vector output $\boldsymbol{z}_{\theta}(\boldsymbol{\xi}) = \left[ z_{1}(\boldsymbol{\xi}), \ldots, z_{L}(\boldsymbol{\xi}) \right] \in \mathbb{R}^{L}$, where $z_{l}(\boldsymbol{\xi})$ represents the raw logit output for action $l \Delta \mu$ given state $\boldsymbol{\xi}$. We adopt the standard \ac{PPO} training procedure as described in \cite{ppo} and focus on the integration of the action mask mechanism. Typically, the policy in standard \ac{PPO} is derived by applying a softmax over the vector output $\boldsymbol{z}_{\theta}(\boldsymbol{\xi})$, as follows:
\vspace{-4pt}
\begin{equation}
    \pi_\theta(l \Delta \mu \mid \boldsymbol{\xi})=\frac{\exp \left(z_{l}(\boldsymbol{\xi})\right)}{\sum_{l'=1}^L \exp \left(z_{l'}(\boldsymbol{\xi})\right)},
\end{equation}
which is essentially the probability of selecting $\Delta\tilde{\mu} = l \Delta \mu$ given state $\boldsymbol{\xi}$. Due to the state-dependent action, we aim to modify the probability of selecting invalid actions to zero. To this end, the action mask mechanism adjusts the logit output $z_{l}(\boldsymbol{\xi})$ as following:
\vspace{-4pt}
\begin{equation}
\label{adjust logit}
    \tilde{z}_l(\boldsymbol{\xi})=z_l(\boldsymbol{\xi})+\log \left(\alpha+(1-\alpha) m_l(\boldsymbol{\xi})\right), 
\end{equation}
where $\alpha \in (0,1)$ is a small positive constant (e.g., $0.01$) and $m_l(\boldsymbol{\xi})$ is the binary action mask defined as
\vspace{-4pt}
\begin{equation}
\label{action mask}
    m_l(\boldsymbol{\xi})  =
\begin{cases}
    1, 
    &  l \Delta \mu \in \mathcal{A}(\boldsymbol{\xi}), \\
    0,
    & l \Delta \mu \notin \mathcal{A}(\boldsymbol{\xi}).
\end{cases}
\end{equation}
Subsequently, the masked policy is redefined as
\vspace{-4pt}
\begin{equation}
\label{modified policy}
    \hat{\pi}_\theta(l \Delta \mu \mid \boldsymbol{\xi})=\frac{\exp \left(\hat{z}_{l}(\boldsymbol{\xi})\right)}{\sum_{l'=1}^L \exp \left(\hat{z}_{l'}(\boldsymbol{\xi})\right)}.
\end{equation}
Based on \eqref{adjust logit}, \eqref{action mask}, and \eqref{modified policy}, we can observe that the masked policy assigns lower probabilities to invalid actions, reducing their likelihood of being selected to near zero. Given the modified policy in \eqref{modified policy}, the influence of the binary action mask $\boldsymbol{m}(\boldsymbol{\xi}) = \left[ m_1(\boldsymbol{\xi}), \ldots, m_L(\boldsymbol{\xi})\right]$ remains differentiable, thereby ensuring faster and more stable convergence of the \ac{PPO} training process and a valid sensing spoofing in our scenario. The overall \ac{PPO} framework with action mask is shown in Fig.~\ref{PPO_action_mask}.
\begin{figure}[t]
	\centering
	\includegraphics[scale=0.32]{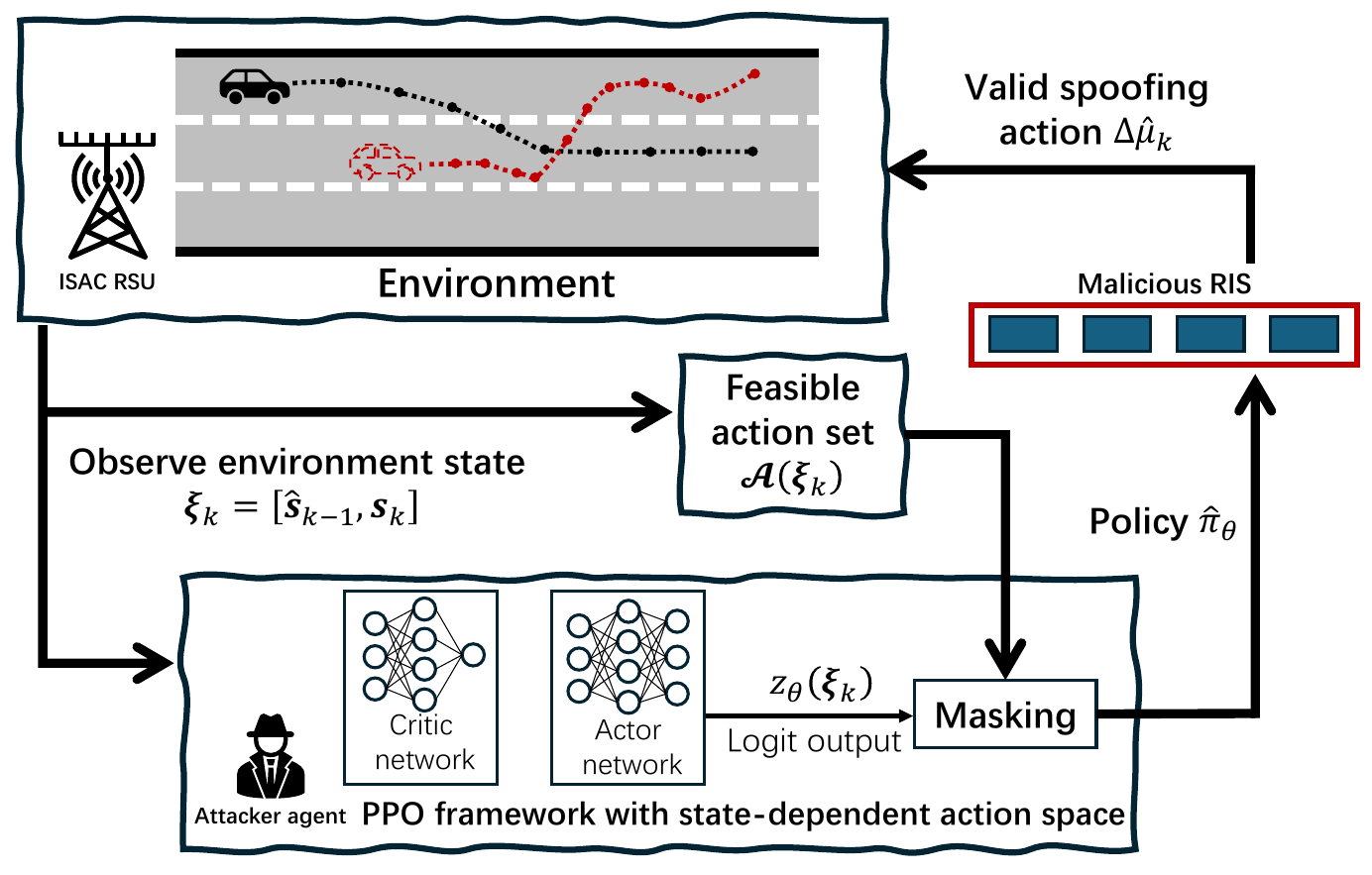}
      \vspace{-10pt}
	\caption{\small{The \ac{PPO} framework with action mask for spoofing the trajectory of \ac{VU}.}}
 \label{PPO_action_mask}
   \vspace{-18pt}
\end{figure}

In presence of trajectory-level sensing spoofing, it is challenging to identify spoofed trajectories via typical deep clustering-based detection method \cite{ruff2018deep}, which extracts the latent representation of a given trajectory. By measuring the distance between the obtained representation and the cluster center of regular trajectories' representations, the spoofed trajectories can be identified. However, since the spoofed trajectory induced by \eqref{opt1} still follows the consistency constraints, its latent representation may not be separable from those of regular trajectories. To address this issue, we observe that real-world trajectories generally exhibit certain intentions (e.g., lane changing), which are difficult for an attacker to imitate under \eqref{c1-1}. Thus, we next propose an \ac{STL}-based attack detection framework which learns interpretable formulas that characterize how regular trajectory evolves over time according to typical driving intentions.

\vspace{-8pt}
\section{STL-based Attack Detection Framework}
\label{Section V}
In this section, from a defender's perspective, we design an attack detection framework at the \ac{RSU}. Specifically, we incorporate a neuro-symbolic structure named TLINet \cite{TLINet} to learn the \ac{STL} formulas for identifying spoofed trajectories.

\vspace{-6pt}
\subsection{\ac{STL} Preliminaries }
Briefly, \ac{STL} is a formal language expressing the complex temporal and logical properties of a signal sequence. Through a series of structural operators, \ac{STL} can quantify how well a signal sequence satisfies some desired objectives. Here we only consider a simplified \ac{STL} without the until operator \cite{Maler2004MonitoringTP}. The syntax and semantics of \ac{STL} are defined as follows:
\begin{definition}
\label{definition 1}
For a trajectory $\boldsymbol{S} = [\boldsymbol{s}_1, \ldots, \boldsymbol{s}_K]$ over $K$ time slots, the \emph{syntax} of the corresponding \ac{STL} formulas is defined recursively as:
\vspace{-4pt}
\begin{equation}
\phi ::= \mu \mid \phi_1 \land \phi_2 \mid \phi_1 \lor \phi_2 \mid \Diamond_{[k_1, k_2]}\phi \mid \Box_{[k_1, k_2]}\phi,
\end{equation}
\noindent
where $\mu := f_\mu(\boldsymbol{s}_k) \geq 0$ is a Boolean predicate with $f_\mu : \mathbb{R}^d \to \mathbb{R}$. $\phi, \phi_1, \phi_2$ are STL formulas. The Boolean operators $\land, \lor$ are conjunction and disjunction, respectively. The temporal operators $\Diamond, \Box$ represent ``eventually" and ``always". Let $0 \leq k_1 \leq k_2 \leq K$ and $k_1, k_2 \in \mathbb{Z}$, $\Diamond_{[k_1, k_2]}\phi$ is true if $\phi$ is satisfied for at least one point $k \in [k_1, k_2] \cap \mathbb{Z}$, while $\Box_{[k_1, k_2]}\phi$ is true if $\phi$ is satisfied for all time points $k \in [k_1, k_2] \cap \mathbb{Z}$.
\end{definition}
\noindent
Definition \ref{definition 1} provides a formal way to describe the temporal and logical properties of a trajectory $\boldsymbol{S}$. Compared with the simple consistency requirement (e.g., \eqref{consistency}), such \ac{STL} formulas capture the inherent temporal patterns of $\boldsymbol{S}$, providing a more semantic and interpretable perspective in identifying whether a certain trajectory is spoofed. To quantify how well a trajectory $\boldsymbol{S}$ satisfies the \ac{STL} formula $\phi$, the quantitative semantics of $\phi$, also called the \ac{STL} robustness, is defined next. 
\begin{definition}
\label{definition 2}
The \emph{quantitative semantics} of an \ac{STL} formula $\phi$ for $\boldsymbol{s}_k$ are defined as:
\begin{align}
\vspace{-4pt}
r(\boldsymbol{s}_k, \mu) &= f_\mu(\boldsymbol{s}_k), \label{ra} \\
r(\boldsymbol{s}_k, \wedge_{i=1}^{n} \phi_i) &= -\max_{i=1,\dots,n} \left[ - r(\boldsymbol{s}_k, \phi_i) \right], \label{rc} \\
r(\boldsymbol{s}_k, \vee_{i=1}^{n} \phi_i) &= \max_{i=1,\dots,n} r(\boldsymbol{s}_k, \phi_i), \label{rd} \\
r(\boldsymbol{s}_k, \Box_{[k_1,k_2]} \phi) &= - \max_{k' \in [k + k_1, k + k_2]} \left[ - r(\boldsymbol{s}_{k'}, \phi) \right], \label{re} \\
r(\boldsymbol{s}_k, \Diamond_{[k_1,k_2]} \phi) &= \max_{k' \in [k + k_1, k + k_2]} r(\boldsymbol{s}_{k'}, \phi). \label{rf}
\end{align}
\end{definition}
\noindent
With Definition \ref{definition 2}, a trajectory $\boldsymbol{S}$ is said to satisfy $\phi$, i.e., $\boldsymbol{S} \models \phi$, if and only if $r(\boldsymbol{s}_0, \phi) \geq 0$. Otherwise, $\boldsymbol{S}$ is said to violate $\phi$, i.e., $\boldsymbol{S} \nvDash \phi$ with $r(\boldsymbol{s}_0, \phi) < 0$. For example, if $r(\boldsymbol{s}_k, \mu) \geq 0$, we can conclude from \eqref{ra} that the Boolean predicate $\mu := f_\mu(\boldsymbol{s}_k) \geq 0$ is true since $r(\boldsymbol{s}_k, \mu) = f_\mu(\boldsymbol{s}_k) \geq 0$. 
Similarly, for \eqref{rc}, if $r(\boldsymbol{s}_k, \wedge_{i=1}^{n} \phi_i) = -\max_{i=1,\dots,n} \left[ - r(\boldsymbol{s}_k, \phi_i) \right] \geq 0$, we can derive that $r(\boldsymbol{s}_k, \phi_i) \geq 0, \forall i= 1, \ldots, n$, which means $\boldsymbol{s}_k$ satisfies all the formulas $\phi_1, \ldots, \phi_n$, i.e., $\wedge_{i=1}^{n} \phi_i$. The rest quantitative semantics \eqref{rd} - \eqref{rf} can be explained in a similar way. Given Definitions \ref{definition 1} and \ref{definition 2}, \ac{STL} formulas can be used by the \ac{RSU} to determine whether the sensed trajectory $\hat{\boldsymbol{S}}$ satisfies certain properties, thereby enabling the detection of potential \ac{RIS} spoofing attacks.

Our goal is to learn an \ac{STL} formula $\phi$ that can accurately distinguish fake trajectories from regular, real-world trajectories, which is essentially an \ac{STL} formula inference problem \cite{10156357}. Specifically, assume that we have a  
labeled trajectory dataset $\mathcal{D} = \left\{ (\boldsymbol{S}_1, c_1) \ldots, (\boldsymbol{S}_D, c_D) \right\}$ with $c_d = 1$ indicating normal trajectory $\boldsymbol{S}_d$ collected without spoofing while $c_d = 0$ indicating abnormal trajectory $\boldsymbol{S}_d$ collected under spoofing. Then, the objective is to learn both structure and parameters (such as $k_1$, $k_2$, and $f_\mu(\cdot)$) of $\phi$ to minimize the misclassification rate $\frac{1}{D} \left| \left\{ \boldsymbol{S}_d \mid \left( \boldsymbol{S}_d \models \phi \wedge c_d = 1 \right) \vee \left( \boldsymbol{S}_d \nvDash \phi \wedge c_d = 0 \right) = 1 \right\} \right|$. Such an \ac{STL}-based attack detection framework can be easily deployed, as it only requires a dataset of vehicle trajectories collected by the \ac{RSU}. These data can be practically obtained via the \ac{RSU}'s sensing functionality and used to learn the \ac{STL} formulas offline, independently of the detection phase. Finally, the learned formula $\phi$ can identify \ac{RIS} spoofing attack once $r(\hat{\boldsymbol{s}}_0, \phi) < 0$ is detected for a sensed trajectory $\hat{\boldsymbol{S}}$.

\vspace{-5pt}
\subsection{Clustering of One-class Dataset}
However, in practice, it is challenging to obtain a labeled data set $\mathcal{D}$ that includes enough spoofed outcomes with $c_d = 0$ because of two reasons: (i) It is easy to obtain normal trajectory since all sensing outcomes collected without spoofing can be defined as normal, while abnormal trajectory is hard to define; (ii) The samples of abnormal trajectories collected under attacks are comparatively limited. Thus, we are faced with a one-class anomaly detection problem \cite{ruff2018deep} where we need to classify normal and abnormal trajectories, while only trajectories labeled as normal are available.

To address this challenge, we propose a self-supervised framework. Specifically, our approach, inspired by pseudo-labeling methods \cite{lee2013pseudo}, first clusters the unlabeled dataset $\mathcal{D}$ into $P$ behavioral categories. For each identified type $p$, a one-vs-rest attack detector is trained based on a new dataset $\mathcal{D}_p$. The trajectory $\boldsymbol{S}_d$ in $\mathcal{D}_p$ is labeled as ``in-class” ($c_d = 1$) if it belongs to the cluster $p$, and ``out-of-class”  ($c_d = 0$) otherwise, serving as pseudo-labels for learning a one-vs-rest attack detector. Subsequently, an \ac{STL} formula $\phi_p$ is learned to characterize the temporal logic pattern for type $p$ and $P$ \ac{STL} formulas $\phi_1, \ldots, \phi_P$ will be obtained in total. During deployment, the \ac{RSU} first classifies $\hat{\boldsymbol{S}}$ into one of the predefined types, e.g., $p$. Then, the corresponding learned \ac{STL} formula $\phi_p$ is used to verify the consistency of $\hat{\boldsymbol{S}}$ through its robustness $r(\hat{\boldsymbol{s}}_0, \phi_p)$. Given $\hat{\boldsymbol{S}}$, sensing attack is detected when $\boldsymbol{S}_d \nvDash \phi$, i.e., $r(\hat{\boldsymbol{s}}_0, \phi_p) < 0$. The overall pipeline for generating multiple one-vs-rest training datasets is illustrated in Fig.~\ref{cluster-figure}.
\begin{figure}[t]
	\centering
	\includegraphics[scale=0.32]{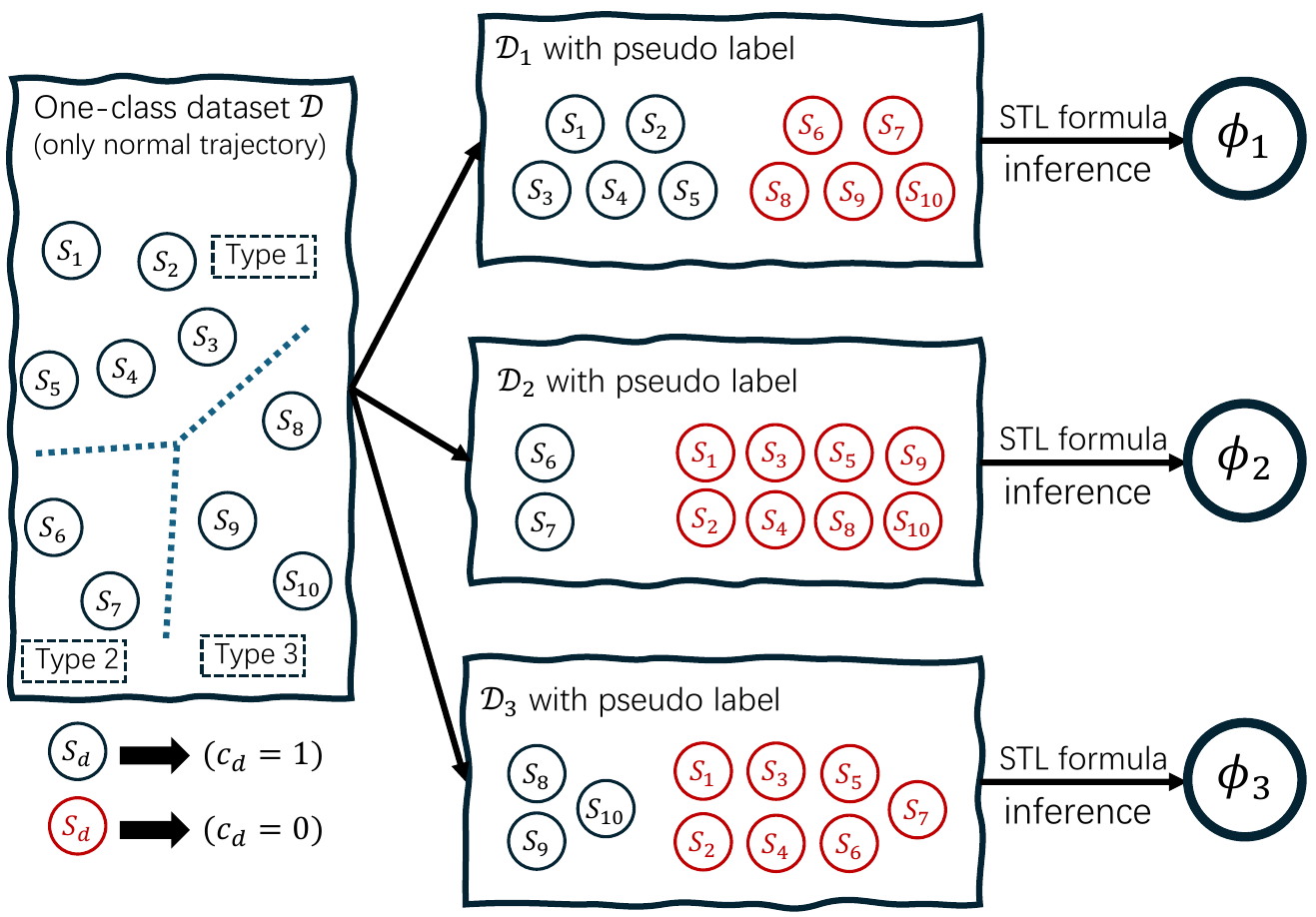}
      \vspace{-8pt}
	\caption{\small{The illustration of clustering a one-class dataset into $P$ pseudo-labeled datasets.}}
    \label{cluster-figure}
      \vspace{-15pt}
\end{figure}

In the proposed framework, we incorporate \ac{DTCR} \cite{NEURIPS2019_1359aa93} to cluster the spatial temporal trajectories in dataset $\mathcal{D}$. Specifically, \ac{DTCR} clusters unlabeled data into multiple categories based on their latent representations. Different from the two-stage approaches that independently extracts latent features and then clusters, \ac{DTCR} jointly optimizes these two process. In short, the latent representations are not only learned for accurate reconstruction, but also for an effective clustering. The construction of $P$ one-vs-rest datasets consists of two components: a) an autoencoder that learns the latent representation of the given trajectory, and b) a K-means clustering algorithm to guide the latent representation learning. 
\subsubsection{Representation Learning}
For the representation learning, we aim to train an encoder $f_{\textrm{En}}(\cdot) : \mathbb{R}^{d \times K} \to \mathbb{R}^m $ that maps the trajectory $\boldsymbol{S}_d$ into latent representations, where $m$ is the dimension of the latent space. The learned representation $\boldsymbol{h}_d = f_{\textrm{En}}(\boldsymbol{S}_d)$ should be able to reconstruct the input trajectory through a paired decoder $f_{\textrm{De}}(\cdot) : \mathbb{R}^m \to \mathbb{R}^{d \times K}$. Specifically, parameterized $f_{\textrm{En}}(\cdot)$ and $f_{d}(\cdot)$ should be trained to minimize the reconstruction loss $L_{\textrm{re}} = \frac{1}{D} \sum_{d=1}^D \left\| f_{\textrm{De}}(\boldsymbol{h}_d) -  \boldsymbol{S}_d \right\|^2$. However, the learned representation may be restricted to the reconstruction task if only the reconstruction loss is considered. To facilitate the latent representation $\boldsymbol{h}_d$ a good clustering performance, features distinct for clustering should be incorporated.
\subsubsection{K-means Clustering}
Given the latent representation matrix $\boldsymbol{H} = \left[ \boldsymbol{h}_1, \ldots, \boldsymbol{h}_D \right] \in \mathbb{R}^{M \times D}$, the standard K-means clustering objective can be equivalently reformulated, following \cite{NIPS2001_d5c18698}, as finding discrete cluster indicator matrix $\boldsymbol{F} \in \mathbb{R}^{D \times P}$ to minimize $L_{\textrm{cl}} = \operatorname{Tr}\left(\boldsymbol{H}^{T} \boldsymbol{H}\right)-\operatorname{Tr}\left(\boldsymbol{F}^{T} \boldsymbol{H}^{T} \boldsymbol{H} \boldsymbol{F}\right)$. The $d$-th row of $\boldsymbol{F}$ indicates the cluster membership of trajectory $\boldsymbol{S}_d$. Since optimizing a discrete indicator matrix $\boldsymbol{F}$ is NP-hard, a common relaxation is to drop the discrete constraint and instead require $\boldsymbol{F}$ to satisfy the orthogonality condition $\boldsymbol{F}^\top \boldsymbol{F} = \boldsymbol{I}$. The closed-form solution of $\boldsymbol{F}$ can be then obtained by composing the first $P$ singular vectors of $\boldsymbol{H}$ according to the Ky Fan theorem \cite{NEURIPS2019_1359aa93}. Finally, to obtain discrete cluster labels, K-means clustering is applied to the row vectors of the relaxed matrix $\boldsymbol{F}$, i.e., a more cluster-friendly representation of $\boldsymbol{H}$, for clustering the trajectories in $\mathcal{D}$.
\subsubsection{Joint Training}
Recall that our aim is to learn $\boldsymbol{H}$ that simultaneously captures essential features for reconstruction and preserves discriminative structures for effective clustering. Thus, the joint training considering both reconstruction loss $L_{\textrm{re}}$ and the clustering loss $L_{\textrm{cl}}$ to optimize $f_{\textrm{En}}(\cdot)$ and $f_{\textrm{De}}(\cdot)$ can be formulated as follows:
\begin{subequations}
\vspace{-8pt}
\label{opt2}
\begin{IEEEeqnarray}{s,rCl'rCl'rCl}
& \underset{\boldsymbol{F}, \boldsymbol{H}}{\text{min}} &\quad&  L_{\textrm{joint}} = L_{\textrm{re}} + \lambda_0 L_{\textrm{cl}}  \label{obj2}\\
\vspace{-2pt}
&\text{s.t.} && \boldsymbol{F}^T \boldsymbol{F}=\boldsymbol{I},  \label{c2-1} 
\end{IEEEeqnarray}
\end{subequations}
where $\lambda_0$ is the regularization parameter. Given fixed $\boldsymbol{F}$, the joint objective \eqref{opt2} can be optimized with respect to the parameters of $f_{\textrm{En}}(\cdot)$ and $f_{\textrm{De}}(\cdot)$ via gradient backpropagation. Therefore, we adopt an alternating optimization strategy, where $\boldsymbol{F}$ and the parameters of $f_{\textrm{En}}(\cdot)$ and $f_{\textrm{De}}(\cdot)$ are updated iteratively. The overall training procedure of the clustering is summarized in Algorithm \ref{algo1}.
\begin{algorithm}[t]
\footnotesize
\caption{Clustering of the one-class dataset}
\label{algo1}
\begin{algorithmic}[1] 
\Require One-class dataset $\mathcal{D}$, number of clusters $P$, iteration round $i_{\textrm{max}}$, and update round  $i_{\textrm{T}}$
\Ensure Encoder $f_{\textrm{En}}(\cdot)$, decoder $f_{\textrm{De}}(\cdot)$, $P$ cluster centers $\boldsymbol{\delta}_1, \ldots, \boldsymbol{\delta}_P$, and $P$ pseudo-labeled datasets $\mathcal{D}_1, \ldots, \mathcal{D}_P$
\State Initialize $f_{\textrm{En}}(\cdot)$ and $f_{\textrm{De}}(\cdot)$ with random parameters.
\State Derive $\boldsymbol{F}$ based on the representation matrix $\boldsymbol{H}$ given by the initial encoder $f_{\textrm{En}}(\cdot)$.
\For{$i = 1$ to $i_{\textrm{max}}$}
    \State With derived $\boldsymbol{F}$, update the parameters of $f_{\textrm{En}}(\cdot)$ and $f_{\textrm{De}}(\cdot)$ according to the loss given in \eqref{obj2} \;
    \If{$i \ \% \ i_{\textrm{T}} = 0$}
        \State With derived $\boldsymbol{H}$, update $\boldsymbol{F}$ according to Ky Fan theorem.
    \EndIf
\EndFor
\State Apply K-means clustering to $\boldsymbol{F}$ for obtaining the cluster membership of trajectories in $\mathcal{D}$ and $P$ cluster centers $\boldsymbol{\delta}_1, \ldots, \boldsymbol{\delta}_P$.
\For{$p = 1$ to $P$}
\State Label all the ``in-class" trajectories $\boldsymbol{S}_d$ in $\mathcal{D}$ clustered as type $p$ with $c_d=1$, while the remaining ``out-of-class" trajectories with $c_d=0$.
\State Construct $\mathcal{D}_p$ with pseudo labels.
\EndFor
\end{algorithmic}
\end{algorithm}

\vspace{-8pt}
\subsection{Neuro-symbolic Learning of \ac{STL} Formulas}
After constructing $P$ one-vs-rest labeled datasets, we train an attack detection model for each dataset $\mathcal{D}_p$. The objective is to infer an \ac{STL} formula $\phi_p$ that characterizes the normal behavior for type $p$, and subsequently identifies spoofed trajectory.
To this end, we incorporate a novel neuro-symbolic framework named TLINet \cite{TLINet} composed of three differentiable layers. Each layer is designed to capture one unique property of the trajectory, which is detailed as follows: 
\begin{figure*}[t]
\vspace{-15pt}
\captionsetup[subfigure]{font=small}  
    \centering
    \subfloat[]{
        \includegraphics[width=0.27\textwidth]{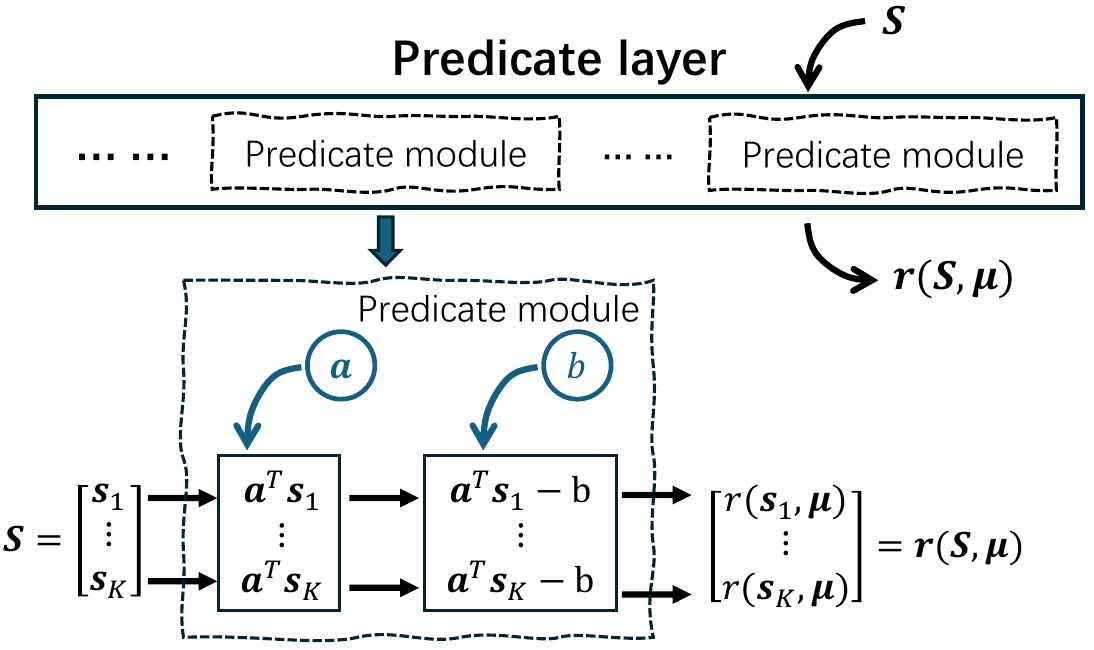}
        \label{predicate_layer}
    }
    \subfloat[]{
        \includegraphics[width=0.30\textwidth]{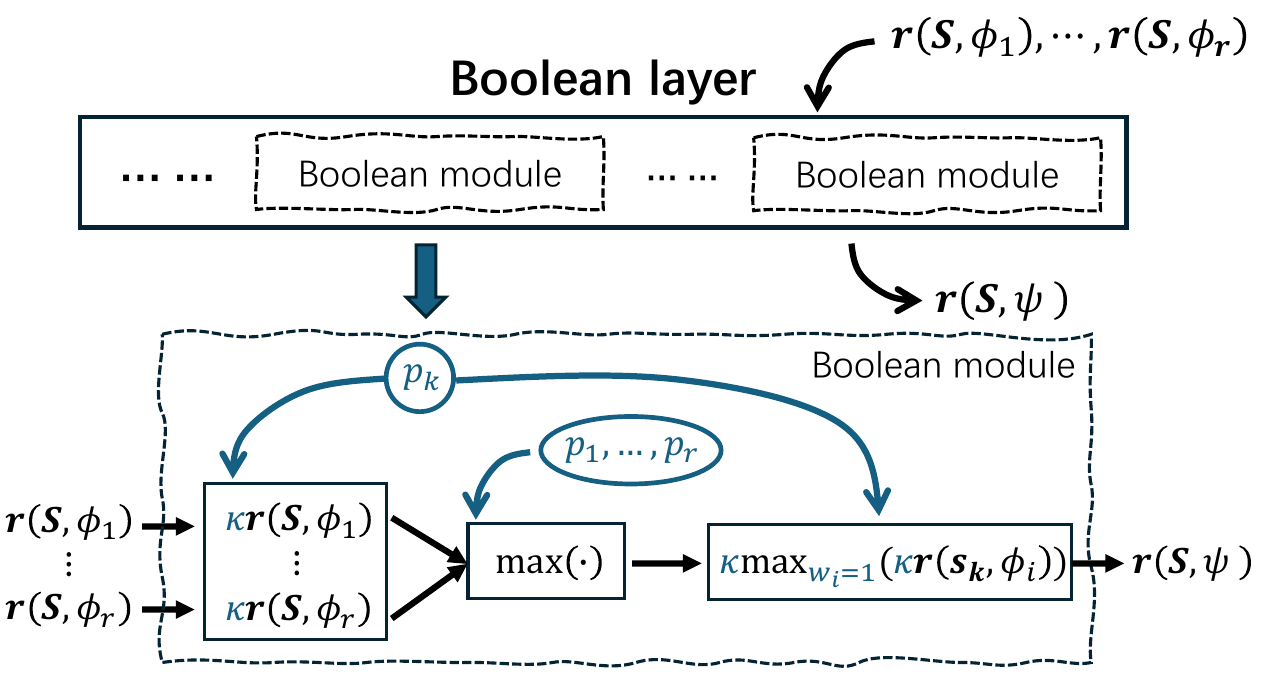}
        \label{boolean_layer}
    }
    \subfloat[]{
        \includegraphics[width=0.34\textwidth]{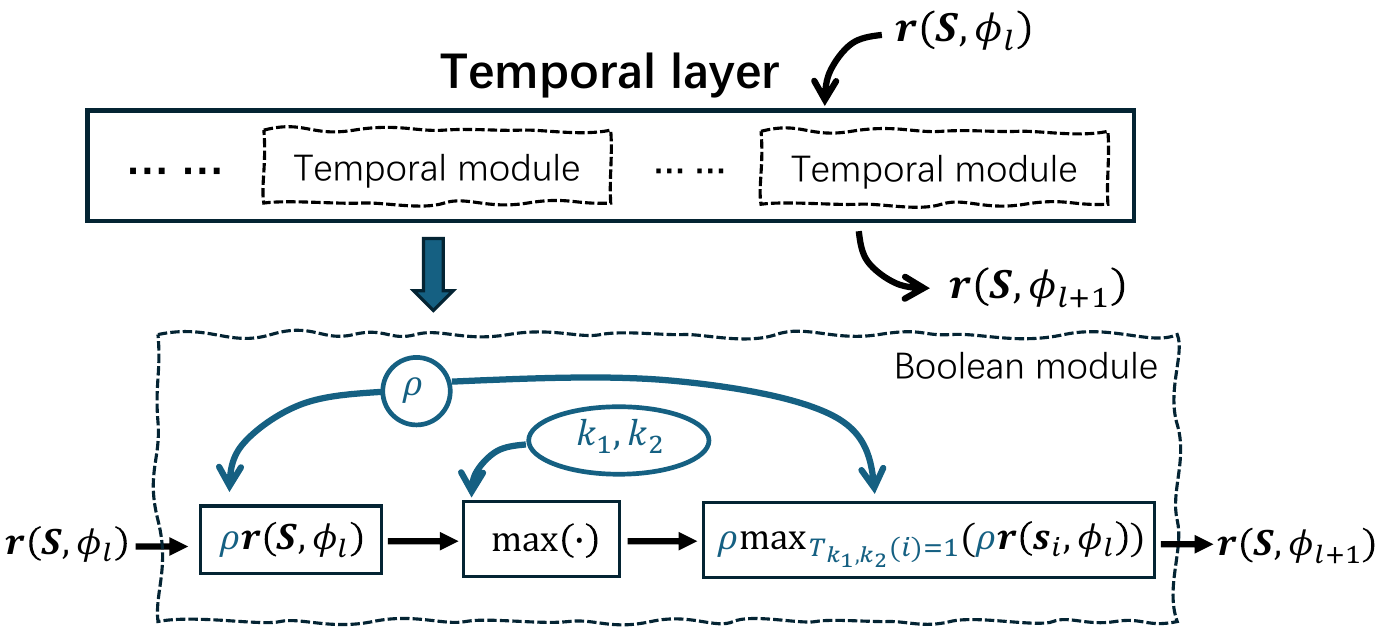}
        \label{temporal_layer}
    }
    \caption{Three basic differentiable layers in TLINet: a) Predicate layer, b) Boolean layer, c) Temporal layer. The trainable parameters are denoted in blue.}
    \label{fig:example}
  \vspace{-18pt}
\end{figure*}
\subsubsection{Predicate Layer}
A predict layer is composed of multiple independent predicate modules. Each predicate module takes trajectory $\boldsymbol{S}$ as input and outputs the robustness vector $\boldsymbol{r}(\boldsymbol{S}, \mu) = \left[r(\boldsymbol{s}_1, \mu), \ldots, r(\boldsymbol{s}_K, \mu) \right]$ over a predicate $\mu$. 
The predicate module captures general low-level properties of trajectories, such as position or speed constraints, that hold across the entire time horizon.
Similar to \cite{10156357} and \cite{TLINet}, we consider a simplified linear predicate $f_\mu(\boldsymbol{s}_k) = \boldsymbol{a}^{T} \boldsymbol{s}_k - b $, where $\boldsymbol{a} \in \mathbb{R}^d$ and $b \in \mathbb{R}$ are the parameters to be learned. 
The structure of predicate layer is shown in Fig. \ref{predicate_layer}.
\subsubsection{Boolean Layer}
A Boolean layer is also composed of multiple dependent Boolean modules. However, a single Boolean module will take the robustness vectors $\boldsymbol{r}(\boldsymbol{S}, \phi_{1}), \ldots, \boldsymbol{r}(\boldsymbol{S}, \phi_{r})$ of a series of sub-formulas $\phi_{1}, \ldots, \phi_{r}$ as input and outputs $\boldsymbol{r}(\boldsymbol{S}, \psi)$ with $\psi = \bigwedge_{i=1}^r \phi_{r} \mid \bigvee_{i=1}^r \phi_{r}$. 
Essentially, the Boolean layer can encode the conjunction and disjunction operation defined in \eqref{rc} and \eqref{rd}, and, thus, the model can now combine multiple low-level predicates. Such integration enables the system to identify spoofed trajectory over combinations of features, e.g., detecting abnormal patterns when multiple spatial constraints are jointly violated. The structure of Boolean layer is illustrated in Fig. \ref{boolean_layer}. To unify the discrete selection between logical operations $\bigwedge_{i=1}^r \phi_{r}$ and $\bigvee_{i=1}^r \phi_{r}$, a learnable selection operator $\kappa \sim \textrm{Ber}_{\textrm{ML}, \mathcal{X}_\kappa}(p_{\kappa})$ with $\mathcal{X} = \left\{ 1, -1 \right\}$ is incorporated. Specifically, the operation of $\bigwedge_{i=1}^r \phi_{r}$ and $\bigvee_{i=1}^r \phi_{r}$ on $\boldsymbol{s}_k$ can be unified as $\kappa \max_{i}(\kappa r(\boldsymbol{s}_k, \phi_{i}))$ according to \eqref{rc} and \eqref{rd}. Moreover, $\textrm{Ber}_{\textrm{ML}, \mathcal{X}}(p)$ with $\mathcal{X} = \left\{ X_0, X_1 \right\}$ is the \textit{maximum likelihood draw} distribution and $x \sim \textrm{Ber}_{\textrm{ML}, \mathcal{X}}(p)$ indicates that
\begin{equation}
\begin{aligned}
&P\left(x=X_0\right)=1, \quad \text { if } 0.5 \leq p \leq 1, \\
\vspace{-4pt}
&P\left(x=X_1\right)=1, \quad \text { if } 0 \leq p<0.5 .
\end{aligned}
\end{equation}
Thus, the discrete selection between different operations can be learned via the training update of parameter $p_{\kappa}$. When $0.5 \leq p_{\kappa} \leq 1$ is learned, then the operator is $\kappa = 1$, which represents the disjunction operation. Similarly, we define $\boldsymbol{w} = [ w_1, \ldots, w_r ]$ with $w_i \sim \textrm{Ber}_{\textrm{ML}, \mathcal{X}_w}(p_i), \forall i = 1, \ldots, r$ and $\mathcal{X}_w = \left\{ 0, 1 \right\}$ to guide the participation of sub-formulas $\phi_{1}, \ldots, \phi_{r}$ in Boolean operation. Thus, $\phi_r$ will only be included in the Boolean composition if $w_r = 1$, allowing the network to selectively compose sub-formulas via the training update of parameters $p_1, \ldots, p_r$. The $k$-th element of $\boldsymbol{r}(\boldsymbol{S}, \psi)$ is calculated as $r(\boldsymbol{s}_k, \psi) = \kappa \max_{w_i=1}(\kappa r(\boldsymbol{s}_k, \phi_{i}))$. Finally, the trainable parameters of a Boolean layer are $p_{\kappa}$ and $p_1, \ldots, p_r$. Moreover, to facilitate a differentiable max operation, the averaged max approximation \cite{TLINet} is incorporated.

\subsubsection{Temporal Layer}
The temporal layer is composed of multiple temporal module. Each temporal module takes $\boldsymbol{r}(\boldsymbol{S}, \phi_{l})$ of a lower-level formula $\phi_{l}$ as input and outputs $\boldsymbol{r}(\boldsymbol{S}, \phi_{l+1})$ of a higher-level formula $\phi_{l+1}$ by applying a temporal operator $\Diamond$ or $\Box$ over a specified time interval $[k_1, k_2]$ on $\phi_{l}$. In other words, temporal layer can encode nested temporal operation $\phi_{l+1} = \Diamond_{[k_1, k_2]}\phi_{l} \mid \Box_{[k_1, k_2]}\phi_{l}$. This layer allows the model to identify persistent or transient violations of expected behaviors across time. For instance, trajectories with repeated violations of a given predicate over multiple time steps are generally more indicative of spoofing. Thus, the temporal layer enables the system to capture long-range temporal consistency in the trajectory of a \ac{VU}. Similarly, a learnable selection operator $\rho \sim \textrm{Ber}_{\textrm{ML}, \mathcal{X}_\rho}(p_{\rho})$ with $\mathcal{X}_\rho = \left\{ 1, -1 \right\}$ is used in the temporal layer to determine different logical operations $\Diamond_{[k_1, k_2]}\phi_{l}$ and $\Box_{[k_1, k_2]}\phi_{l}$ (\eqref{re} and \eqref{rf}). To facilitate the training of the time interval $[k_1, k_2]$, a trainable time vector $T_{k_1, k_2} = [i_0, \ldots, i_{K-1}] \in \left\{0,1\right\}^K$ is incorporated:
\vspace{-5pt}
\begin{equation}
    \begin{aligned}
T_{k_1, k_2}
= & \frac{1}{\eta} \min \left(\operatorname{ReLU}\left(\mathbf{n}-\mathbf{1}\left(k_1-\eta\right)\right)-\operatorname{ReLU}\left(\mathbf{n}-\mathbf{1} k_1\right) \right. \\
& \left. \times \operatorname{ReLU}\left(-\mathbf{n}+\mathbf{1}\left(k_2+\eta\right)\right)-\operatorname{ReLU}\left(-\mathbf{n}+\mathbf{1} k_2\right)\right),
\end{aligned}
\end{equation}
where $\boldsymbol{n} = [0, 1, \ldots, k-1]$, $\boldsymbol{1}$ is a all one vector, and $\eta$ is a positive hyperparameter controlling the slope steepness of $T_{k_1, k_2}$. $T_{k_1, k_2}$ acts as an indicator vector such that $T_{k_1, k_2}(k) = 1, k_1 \leq k \leq k_2$ and $T_{k_1, k_2}(k) = 0, k < k_1 \vee k > k_2$. Consequently, the temporal operators $\Diamond_{[k_1, k_2]}\phi_{l}$ and $\Box_{[k_1, k_2]}\phi_{l}$ operate exclusively over time steps with indicator vector $T_{k_1, k_2}$ equal to 1, restricting their scope to the interval $ k_1 \leq k \leq k_2$. Element $k$ of $\boldsymbol{r}(\boldsymbol{S}, \phi_{l+1})$ is calculated as $r(\boldsymbol{s}_k, \phi_{l+1}) = \rho \max_{T_{k_1, k_2}(i)=1}(\rho r(\boldsymbol{s}_i, \phi_{l}))$. The trainable parameters of the Boolean module are $p_{\rho}$ and $k_1, k_2$, as shown in Fig. \ref{temporal_layer}. Similarly, the differentiable max operation is realized by the sparse softmax proposed in \cite{TLINet}.

\begin{figure}[t]
	\centering
	\includegraphics[scale=0.34]{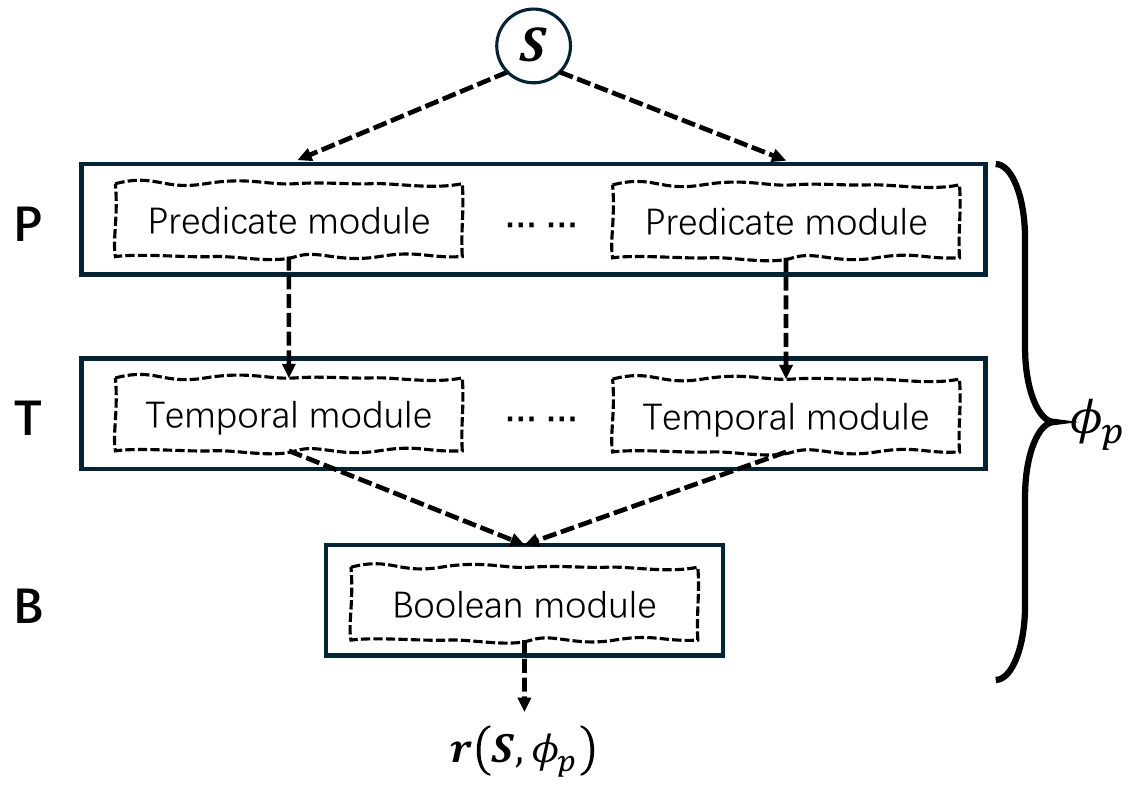}
       \vspace{-6pt}
	\caption{\small{An illustrative example of TLINet composed of one predicate layer, one temporal layer, and one Boolean layers.}}
       \vspace{-12pt}
 \label{layer_structure}
\end{figure}
By flexibly stacking temporal or Boolean layers on top of the predicate layer, we can construct TLINet to learn the \ac{STL} formula $\phi_p$ for each dataset $\mathcal{D}_p$. A representative illustration is given in Fig.~\ref{layer_structure}. Here, the predicate layer, serving as the foundation, learns basic predicate of the trajectory's features. Then, a temporal layer is stacked to formulate the temporal relationship of the trajectory. Finally, a Boolean layer is used to decide the conjunction or disjunction logic and output the robustness value of the given trajectory. For the training process of the \ac{STL} formula, a composite loss function incorporating task-specific objectives and regularization terms is adopted. Take $\phi_p$ for instance, the total loss is given by
\vspace{-6pt}
\begin{equation}
    L_{\textrm{STL}}=\mathcal{L}_{\text {task }}+\lambda_1 \mathcal{L}_s+\lambda_2 \mathcal{L}_{\text {avm }}+\lambda_3 \mathcal{L}_{\text {kavm }},
\end{equation}
where $\mathcal{L}_{\text {task }} = \frac{1}{D} \sum_{d=1}^N e^{-c_d r(\boldsymbol{s}_{d,0}, \phi_p)}$ is the main task loss for classifying ``in-class" and ``out-of-class" samples, $\lambda_1, \lambda_2, \lambda_3 \in \mathbb{R}_{+}$ are balancing coefficients, and $\mathcal{L}_s$, $\mathcal{L}_{\text {avm }}$, $\mathcal{L}_{\text {kavm }}$ are auxiliary regularization losses. In particular, $\mathcal{L}_s$ encourages the sparsity in the learned structure of $\phi_p$, $\mathcal{L}_{\text {avm }}$ serves as the regularizer to promote distinct sub-formula selection in the Boolean layer, and $\mathcal{L}_{\text {kavm }}$ stabilizes choice between different logical operations when applying differentiable logical aggregation. The detailed formulation of each regularization term can be found in \cite{TLINet}. Finally, the \ac{STL}-based attack detection framework is given in Algorithm \ref{algo2}. In summary, the proposed \ac{STL}-based detection framework yields interpretable formulas that enable explicit inference of spoofing attack, while also providing structural transparency of the detection method in safety-critical \ac{ISAC} systems.
\begin{algorithm}[t]
\footnotesize
\caption{\Ac{STL}-based attack detection}
\label{algo2}
\begin{algorithmic}[1] 
\Require Sensing outcome $\boldsymbol{\hat{S}}$ (trajectory), well-trained encoder $f_{\textrm{En}}(\cdot)$, $P$ \ac{STL} formulas $\phi_1, \ldots, \phi_P$, and $P$ cluster centers $\boldsymbol{\delta}_1, \ldots, \boldsymbol{\delta}_P$.
\Ensure Spoofing detection result for $\hat{\boldsymbol{S}}$.
\State Obtain the latent representation of the sensing outcome $\hat{\boldsymbol{h}} = f_{\textrm{En}}(\boldsymbol{\hat{S}})$.
\State Derive the distance between $\hat{\boldsymbol{h}}$ and cluster center $d_p = \left\| \hat{\boldsymbol{h}} - \boldsymbol{\delta}_p \right\|_2$ and cluster $\boldsymbol{\hat{S}}$ into type $p^* = \arg \min_p d_p $.
\State Identify spoofed trajectory if $r(\boldsymbol{\hat{s}}_0, \phi_{p^*})<0$.
\end{algorithmic}
\end{algorithm}

\vspace{-8pt}
\section{Simulation Results and Analysis}
\label{Section VI}
For our simulations, we consider a two-dimensional Cartesian coordinate system with the \ac{RSU} located at its origin. The \ac{RIS} is fixed at $(x_{\textrm{R}}, y_{\textrm{R}}) = $ ($5$~m, $15$~m), i.e., $\theta_{\textrm{R}} = 71.6^{\circ}$. The parameters of the \ac{ISAC} vehicular network are set as follows unless specified otherwise later: $P = 30$~dBm, $\sigma^2 = -100$~dBm, $f_c = 28$~GHz, $N_t = N_r = 32$, $M = 32$, $\kappa_{\textrm{V}} = 7$~dBsm, $\eta = 0.8$, $S = 50~\text{cm} \times 10~\text{cm}$, $T = 10$~ms, and $\Delta T = 1$~ms \cite{10740590,9947033,9171304}
. The parameters for the consistency constraints are $a_{\textrm{max}} = 3~\text{m/s}^2$, $a_{\textrm{max}} = -3~\text{m/s}^2$, $\delta_x = 1$~m, and $\delta_y = 0.3$~m. The hyperparameters for the clustering Algorithm \ref{algo1} are $i_{\textrm{max}} = 50$ and $i_{\textrm{T}} = 5$.

\begin{figure}[t]
	\centering	\includegraphics[scale=0.44]{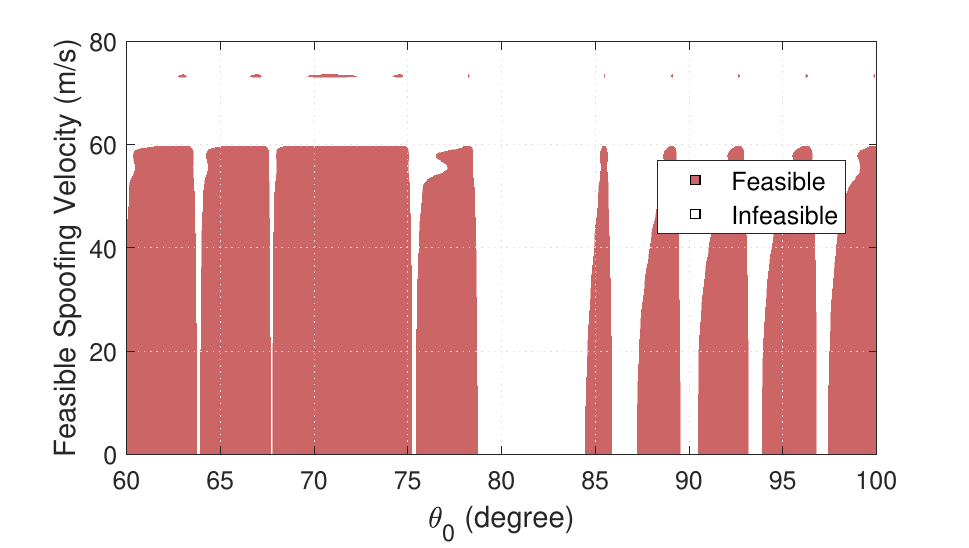}
    \vspace{-8pt}
	\caption{\small{Feasible spoofing velocity versus beam steering direction.}}
    \vspace{-15pt}
    \label{feasible_spoofing_velocity}
\end{figure}
We first examine the impact of the malicious \ac{RIS} on the sensing outcomes in a single time slot. The coordinates of the \ac{VU} in the considered slot are assumed as $(x_{\textrm{V}}, y_{\textrm{V}}) = $ ($3$~m, $21$~m), i.e., $\theta_{\textrm{V}} = 81.9^{\circ}$. Moreover, we assume $v = 10$~m/s, and, thus, the Doppler shift can be derived as $\mu_{\textrm{V}} = v f_c \cos{\theta_{\textrm{V}}}$. Fig. \ref{feasible_spoofing_velocity} shows the impact of \ac{RIS} spoofing on the estimated velocity of the \ac{VU}, which can be derived by the spoofed Doppler shift. In particular, the different feasible spoofing velocity sets, derived from \eqref{feasible set}, versus the \ac{RSU}'s beam steering direction $\theta_0$ are illustrated. From Fig. \ref{feasible_spoofing_velocity}, we can observe that, generally, a beam steered in the proximity of the \ac{VU}'s direction ($\theta_{\textrm{V}} = 81.9^{\circ}$) can eliminate the impact of \ac{RIS} spoofing. For instance, when $\theta_0 \in (78^{\circ}, 84^{\circ})$, we have $\mathcal{A} = \varnothing$, indicating that no successful spoofing can be conducted. However, when $\theta_0$ is slightly misaligned with the \ac{VU}, e.g., $\theta_0 = 77^{\circ}$ (a deviation of $-4.9^{\circ}$) or $\theta_0 = 85^{\circ}$ (a deviation of $3.1^{\circ}$), the \ac{RSU} may obtain a spoofed velocity ranging from $0.1$~m/s to $60$~m/s, introducing a significant estimation error of $-9.9$~m/s to $50$~m/s with respect to the true velocity $v = 10$~m/s.

\begin{figure}[t]
	\centering	\includegraphics[scale=0.51]{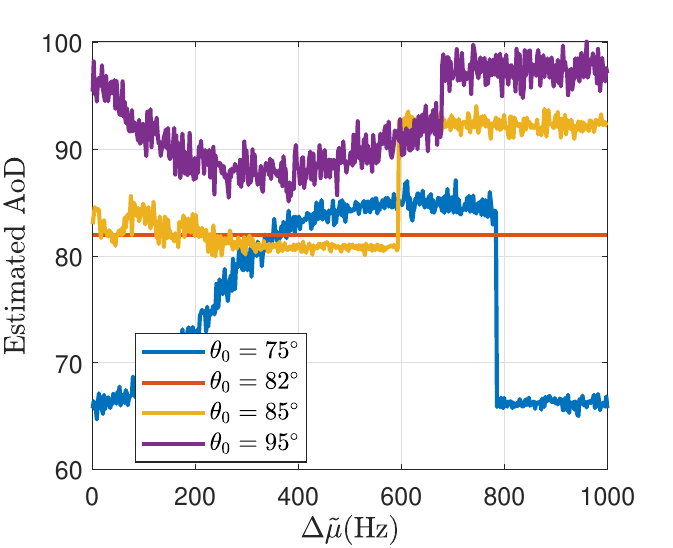}
    \vspace{-8pt}
	\caption{\small{Impact of Doppler shift spoofing on the resulted \ac{MLE} of \ac{AoD}.}}
    \vspace{-12pt}
    \label{spoofed_AoD}
\end{figure}
In Fig. \ref{spoofed_AoD}, we show the spoofed \ac{MLE} for the \ac{AoD} as a function of the spoofing frequency on the Doppler shift estimation. The estimated values for the \ac{AoD} are averaged over \num{200} trials of independent noise. First, we can observe that the \ac{MLE} for the \ac{AoD} will not be affected under a beam perfectly steered to the \ac{VU} when $\theta_0 = 82^{\circ}$. However, the \ac{RIS} can always find a spoofing frequency that causes a severe deviation of the spoofed \ac{MLE} under beam misalignment. For instance, when $\theta_0 = 85^{\circ}$ deviates by only $3.1^{\circ}$ from $\theta_{\textrm{V}}$, its \ac{MLE} can be spoofed as $92^{\circ}$ with $\Delta \tilde{\mu} \geq 600$~Hz, leading to an estimation error of at least $11^{\circ}$. When $\theta_0 = 75^{\circ}$ or $\theta_0 = 95^{\circ}$ is set, the accuracy of \ac{AoD} estimation is even more  significantly compromised.


\begin{figure}[t]
	\centering	\includegraphics[scale=0.36]{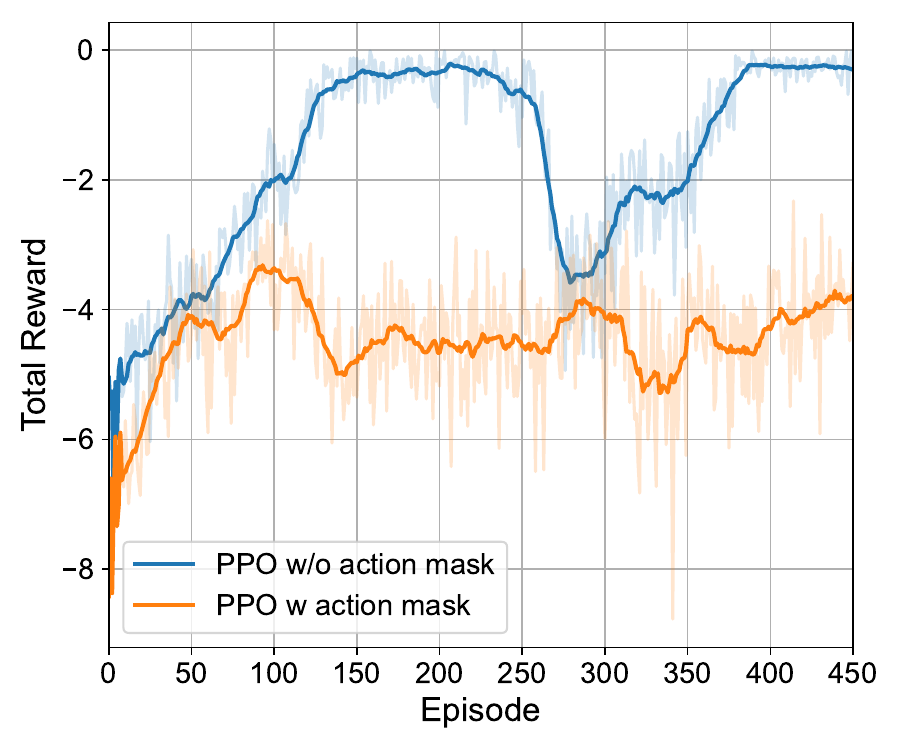}
    \vspace{-6pt}
	\caption{\small{Reward versus \ac{PPO} training episodes.}}
    \vspace{-15pt}
    \label{Reward PPO}
\end{figure}
In Fig. \ref{Reward PPO}, the average reward of the action mask enhanced \ac{PPO} framework is evaluated and compared to the standard \ac{PPO} framework (without action mask). From Fig.~\ref{Reward PPO}, we can observe that the \ac{PPO} framework with action mask, though more stable, counterintuitively achieves lower reward after 400 episodes of training, which seemingly indicates that the action mask mechanism hinders the \ac{RIS} from generating plausible trajectories. However, this observation is misleading and action mask mechanism is actually necessary. Specifically, the \ac{PPO} framework without action mask tends to chose a spoofing frequency whereby $\Delta \tilde{\mu}_k \notin \mathcal{A}_k$, which naturally satisfies the consistency with higher reward and fails to conduct sensing attack (see Theorem \ref{theorem 1}). Only with the guarantee provided by action mask mechanism, the invalid spoofing frequencies can be discarded in the \ac{PPO} framework. As a result, the sensing outcomes of the \ac{VU} will largely deviate from the ground truth, as shown next in an \ac{SAC} use case.

\begin{figure}[t]
\vspace{-6pt}
\scriptsize
\captionsetup[subfigure]{font=small}  
    \centering
    \subfloat[]{
        \includegraphics[width=0.24\textwidth]{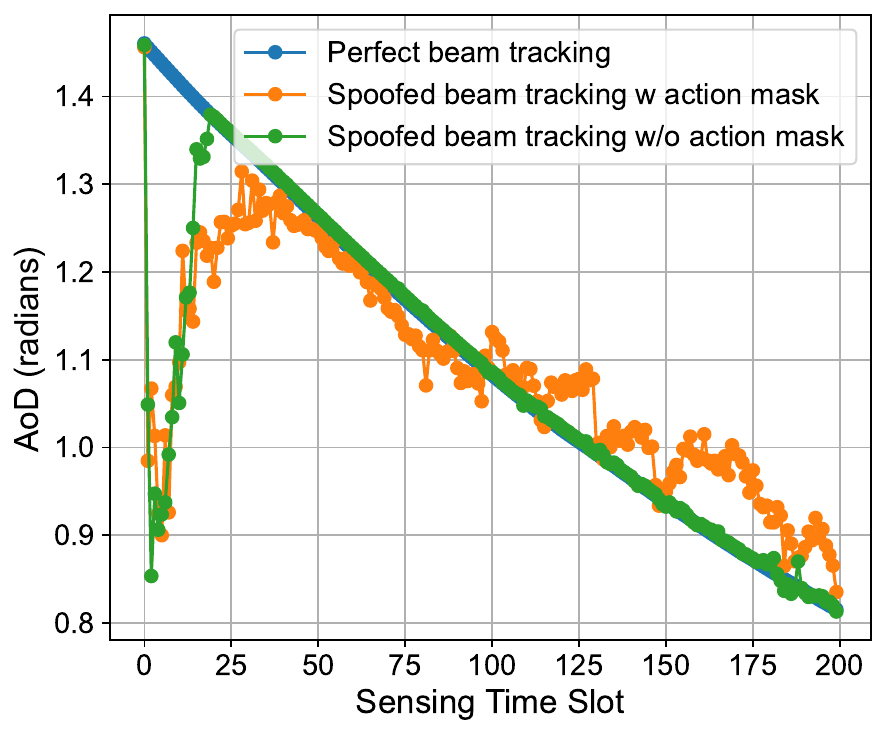}
        \label{AoD traj}
    }
    \subfloat[]{
        \includegraphics[width=0.24\textwidth]{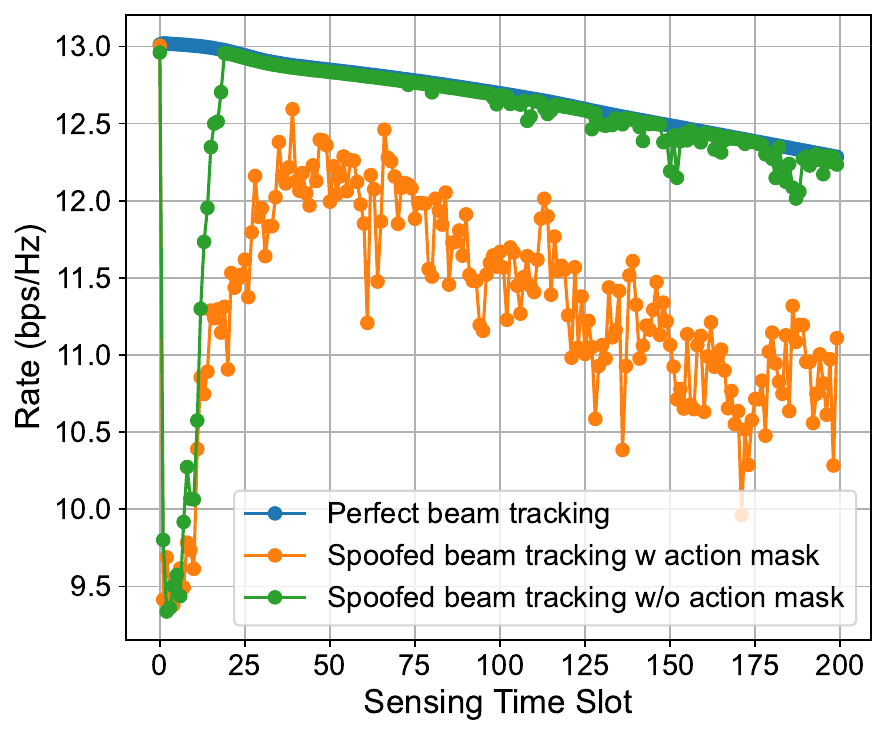}
        \label{Rate traj}
    }
    \vspace{-6pt}
    \caption{Impact of the \ac{RIS} spoofing in the beam tracking application: a) Estimated \ac{AoD}, b) Communication rate. }
    \label{Beam tracking}
        \vspace{-8pt}
\end{figure}
To illustrate the impact of the \ac{RIS} sensing spoofing, we consider the beam tracking application studied in \cite{9171304}. Specifically, the estimated \ac{AoD} and achievable rate of the \ac{VU} are evaluated in Fig. \ref{AoD traj} and Fig. \ref{Rate traj}. Two beam-tracking strategies are considered. Perfect beam tracking represents the scenario in which the \ac{RSU} can obtain an accurate sensing outcomes of the \ac{VU} without the sensing attack, while the spoofed beam tracking represents the scenario under \ac{RIS} spoofing attack. Two \ac{PPO} strategies trained with and without action mask are used to determine the spoofing frequency. From Fig. \ref{AoD traj}, we can observe that both strategies induce \ac{AoD} estimation error up to $0.6$ radian during the first $25$ slots. This is because the \ac{VU} is close to the \ac{RIS}, which leads to a larger feasible frequency set and a less accurate \ac{AoD} estimation. As the \ac{VU} moves away, the strategy trained without action mask fails to spoof the \ac{AoD} estimation while that trained with action mask is still able to result \ac{AoD} estimation error ranging from $-0.1$ to $0.13$ radian. The impact of this \ac{AoD} estimation error in beam tracking is shown in \ref{Rate traj}. For both two strategies, the \ac{VU} experiences an achievable rate decrease of up to $28~\%$ compared to that under the perfect beam tracking in the first $25$ time slots. Moreover, even after the first $25$ time slots, the strategy with action mask can still result an achievable rate decrease from $0.3$~bps/Hz to $2.5$~bps/Hz.

\begin{figure}[t]
	\centering	\includegraphics[scale=0.38]{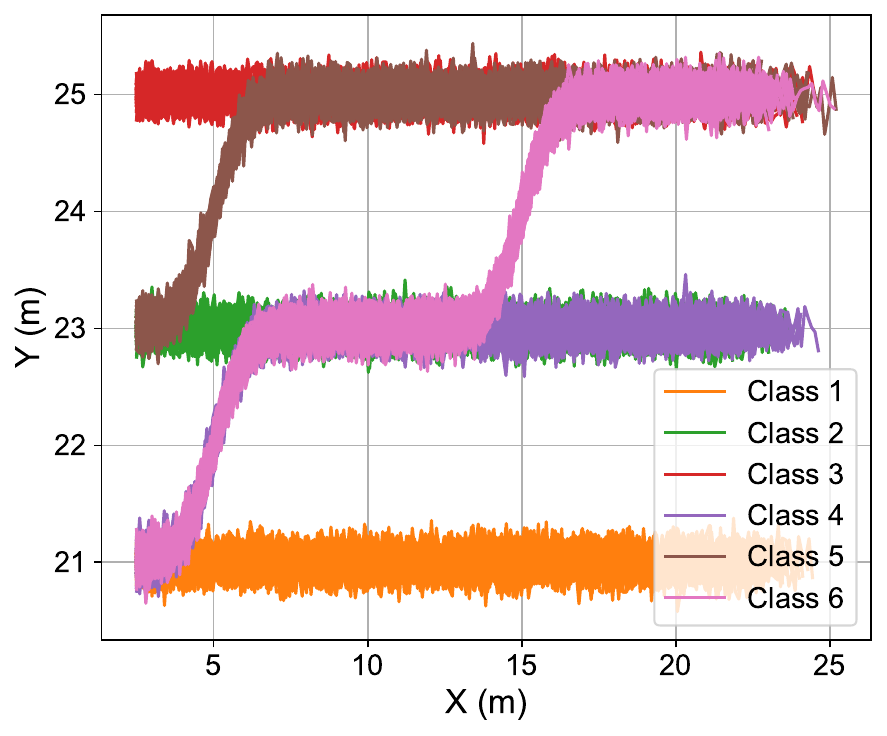}
    \vspace{-6pt}
	\caption{\small{Clustering results of the trajectory dataset.}}
    \vspace{-15pt}
    \label{dataset}
\end{figure}
Next, we evaluate the effectiveness of the proposed \ac{STL}-based neuro-symbolic attack detection framework. First, to train the \ac{DTCR} clustering framework, we consider a one-class dataset composed of $480$ trajectory samples with trajectory length of $67$. By exploiting the \ac{DTCR} clustering framework given in Algorithm 2, the trajectory samples are clustered into six classes, as shown in Fig. \ref{dataset}. Specifically, classes 1, 2, and 3 represent the straight driving patten, classes 4 and 5 represent the single lane changing pattern, and class 6 represents the double lane changing pattern. At the same time, six pseudo-labeled datasets are obtained and each of them will be used to learn a \ac{STL} formula via Algorithm 2. We adopt a TLINet of the same structure as given in Fig. \ref{layer_structure}. The learned formulas for six classes are given in Table \ref{formula} with the training hyperparameters listed in Table \ref{hyper}. To evaluate the clustering framework and the learned formulas, $120$ normal trajectory samples and $120$ spoofed trajectory samples (generated by the trained \ac{PPO} framework with action mask) are used for testing. Moreover, a deep clustering-based attack detection method is incorporated as benchmark. Specifically, the considered benchmark extracts the latent representation of a test trajectory and classifies it as spoofed if the representation lies far from the cluster center of regular trajectories in the one-class dataset. The confusion matrices demonstrating the classification performance of the benchmark model and the proposed detection model are given in Fig. \ref{confusion_matrix_DEEP} and Fig. \ref{confusion_matrix_STL}, respectively, where actual positive and actual negative indicate normal and spoofed trajectory, while predicated positive and negative represent the detection result. From Fig. \ref{confusion_matrix_DEEP}, we can observe that the benchmark fails to identify the spoofed trajectories generated under consistency constraints (with an accuracy of only $20.83 \%$), indicating that their latent representation are not separable from those of regular trajectories. In contrast, our proposed \ac{STL}-based method achieves a significantly higher detection accuracy of $74.17\%$, which is $53.34\%$ higher than that of the deep clustering-based benchmark. Moreover, the learned \ac{STL} formulas, given in Table \ref{formula}, not only provide interpretability but also offer clear logical rules for distinguishing spoofed trajectories.
\begin{table}[t]
\scriptsize
\renewcommand{\arraystretch}{1.5}  
\centering
\caption{Learned \ac{STL} formulas}
\vspace{-5pt}
\label{formula}
    \begin{tabular}{c|c}
    \hline
    \textbf{Class index}    & \textbf{\ac{STL} formula}        \\ \hline
    $1$ & $\square_{[30,31]}0.4111x-0.3976y+3.8745>0.00$ \\ \hline
    $2$ & \makecell[l]{
            $\square_{[35,36]}(0.8856x - 0.6263y + 3.5292 > 0.00)$ \\ 
            $\wedge \Diamond_{[0,1]}(-0.4278x + 0.1899y - 3.2133 > 0.00)$
            }  \\ \hline
    $3$ & $\square_{[3,29]}0.2661x+0.2212y-6.1891>0.00$  \\ \hline
    $4$ & \makecell[l]{
            $\square_{[2,29]}(0.8741x-0.4169y+5.7552>0.00)$ \\
            $\wedge \square_{[32,55]}(0.4738x - 0.1150y - 1.3614>0.00)$
            }  \\ \hline
    $5$ & $\square_{[14,26]}-0.0045x + 0.3247y-7.9649>0.00$ \\ \hline
    $6$ & \makecell[l]{
            $\square_{[1,15]}(0.2649x - 0.3085y + 5.9170>0.00)$ \\
            $\wedge \square_{[5,53]}(-0.7463x + 1.0237y - 10.4846>0.00)$
            }  \\ \hline
    \end{tabular}
    \vspace{-0.5cm}
\end{table}
\begin{table}[t]
\scriptsize
\centering
\caption{Hyperparamters for TLINet}
\vspace{-5pt}
\label{hyper}
    \begin{tabular}{c|ccccc}
    \hline
    \textbf{Class index}    & $\lambda_1$  & $\lambda_2$  & $\lambda_3$ & $\lambda_0$   \\ \hline
    $1$ & $4e-2$ & $5e-0$ & $5e-0$ & $5e1$ \\
    $2$ & $1e-3$ & $5e-0$ & $5e-0$ & $5e1$ \\
    $3$ & $1e-1$ & $1e1$ & $1e1$ & $5e1$ \\
    $4$ & $1e-3$ & $5e-0$ & $5e-0$ & $5e1$ \\
    $5$ & $2e-2$ & $5e-0$ & $5e-0$ & $5e1$ \\
    $6$ & $2e-2$ & $5e-0$ & $5e-0$ & $5e1$ \\ \hline
    \end{tabular}
    \vspace{-0.5cm}
\end{table}
\begin{figure}[t]
\vspace{-6pt}
\scriptsize
\captionsetup[subfigure]{font=small}  
    \centering
    \subfloat[]{
        \includegraphics[width=0.24\textwidth]{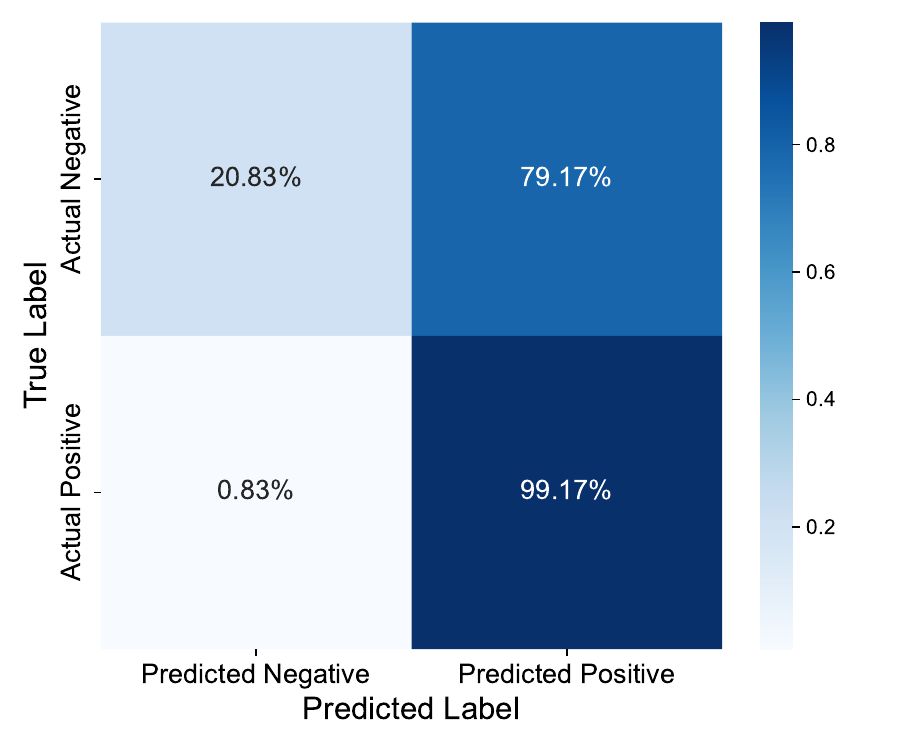}
        \label{confusion_matrix_DEEP}
    }
    \subfloat[]{
        \includegraphics[width=0.24\textwidth]{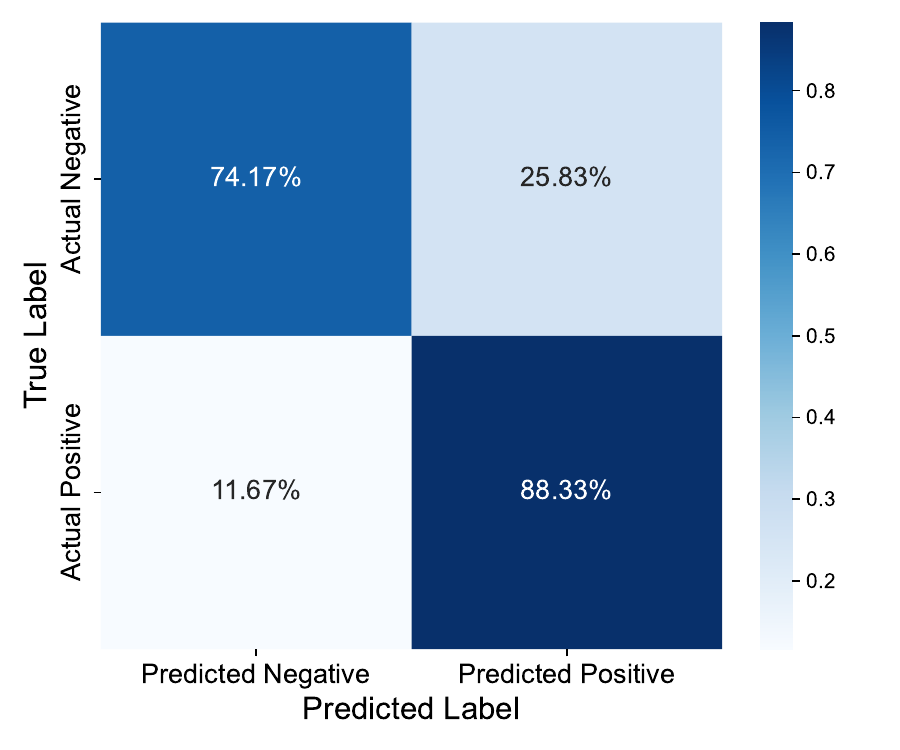}
        \label{confusion_matrix_STL}
    }
    \vspace{-6pt}
    \caption{Confusion matrix of two attack detection frameworks: a) Deep feature-based attack detection, b) Proposed \ac{STL}-based attack detection.}
    \label{confusion_matrix}
        \vspace{-12pt}
\end{figure}

\vspace{-8pt}
\section{Conclusion}
\label{Section VII}
In this paper, we have studied the feasibility of an \ac{RIS} spoofing attack in vehicular \ac{ISAC} network. Specifically, we have analyzed the impact of this spoofing on the Doppler shift and \ac{AoD} estimation of a \ac{VU}, from both slot-level and trajectory-level perspectives. The necessary conditions for the \ac{RIS} to conduct such sensing spoofing are derived. To address this attack, we have proposed an \ac{STL}-based neuro-symbolic attack detection framework, which learns the interoperable formulas for identifying spoofed trajectories.  Simulation results showed that the \ac{RIS} spoofing significantly compromises the sensing accuracy of the \ac{RSU}. Take the beam tracking application as an example, we have further shown that such sensing estimation error further reduces the \ac{VU}'s achievable rate. Finally, the simulation results show that the proposed \ac{STL}-based detection framework improves the detection accuracy by $53.34\%$ over the deep clustering-based benchmark, achieving an accuracy of $74.17\%$ in identifying the \ac{RIS} spoofing.
\vspace{-5pt}
\bibliographystyle{IEEEtran}
\bibliography{bibliography}
\end{document}